\documentclass[journal]{IEEEtran}
\usepackage[T1]{fontenc}
\usepackage{verbatim}
\usepackage{amsmath}
\usepackage{graphicx}
\usepackage[table,xcdraw]{xcolor}
\makeatletter

\usepackage{amsthm}
\usepackage{amssymb} 
\theoremstyle{plain}
\newtheorem{theorem}{Theorem}

\newtheorem{remark}{Remark}
\usepackage{booktabs}
\usepackage{multirow}
\usepackage{array}
\newcolumntype{P}[1]{>{\centering\arraybackslash}p{#1}}
\newcolumntype{C}[1]{>{\centering\arraybackslash}m{#1}}

\def\BibTeX{{\rm B\kern-.05em{\sc i\kern-.025em b}\kern-.08em
    T\kern-.1667em\lower.7ex\hbox{E}\kern-.125emX}}
\usepackage{balance}
\usepackage{cite}
\makeatother

\definecolor{darkcyan}{rgb}{0.0, 0.35, 0.55}

\definecolor{darkgreen}{rgb}{0.0, 0.5, 0.0}

\newcommand{\orcidicon}[1]{\href{https://orcid.org/#1}{\raisebox{0.4ex}{\includegraphics[height=1.6ex]{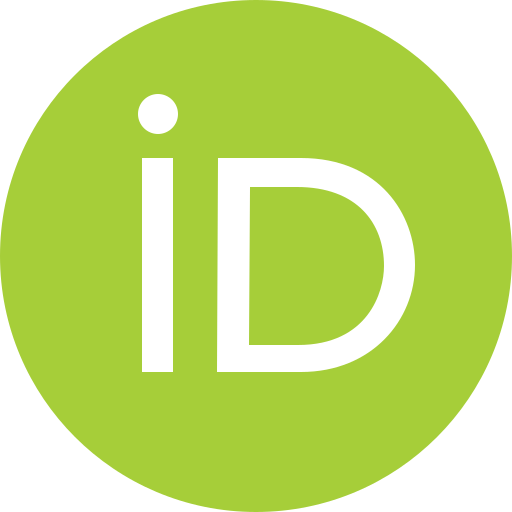}}}}

\usepackage[colorlinks=false, linkcolor=black, citecolor=black, urlcolor=black, pdfborder={0 0 0}, unicode]{hyperref}    

\newcounter{example}[subsection]
\renewcommand{\theexample}{\thesubsection.\arabic{example}}
\newcommand{\example}{\refstepcounter{example}\textit{\textbf{Example \theexample:} }}

\usepackage{fancyhdr}

\fancypagestyle{firstpage}{%
  \fancyhf{} 
  
  \fancyhead[L]{\scriptsize 
 	This work has been submitted to the IEEE for possible publication. Copyright may be transferred without notice, after which this version may no longer be accessible. 
  }%
}
\begin{document}
\title{Joint Sampling Frequency Offset Estimation and Compensation Algorithms Based on the Farrow Structure}
\author{Deijany Rodriguez Linares\orcidicon{0009-0004-1846-9496}, \textit{Graduate Student Member, IEEE}, Oksana Moryakova\orcidicon{0009-0001-6464-5452}, \textit{Graduate Student Member, IEEE}, and H\aa kan Johansson\orcidicon{0000-0001-6329-9132}, \textit{Senior Member, IEEE}%
\thanks{The authors are with the Department of Electrical Engineering, Linköping University, 58183 Linköping, Sweden. Email: \{deijany.rodriguez.linares, oksana.moryakova, hakan.johansson\}@liu.se.}
\thanks{This work was funded by ELLIIT and Sweden's innovation agency.}}

\maketitle
\thispagestyle{firstpage}

\begin{abstract}
	This paper presents joint sampling frequency offset (SFO) estimation and compensation algorithms based on the Farrow structure. Unlike conventional approaches that treat estimation and compensation separately, the proposed framework exploits the interpolator structure to enable a low-complexity, fully time-domain solution applicable to arbitrary bandlimited signals, without imposing constraints on the waveform or requiring Fourier transform based processing. The estimation stage can operate on a real-valued component of a complex signal and supports the simultaneous estimation of SFO and sampling time offset, while being inherently robust to other synchronization impairments such as carrier frequency offset. 
	The proposed estimation algorithms rely on two complementary methods, specifically, Newton's method and iterative least-squares formulation. 
	The implementations of the estimators are presented and the overall computational complexity is analyzed, showing that the complexity scales only linearly with the number of samples employed. 
	Numerical results for real and complex multisine and bandpass-filtered white noise signals demonstrate accurate estimation and effective compensation over a wide range of operating conditions, confirming the flexibility and efficiency of the proposed approach. Moreover, the influence of the Farrow structure approximation error on the SFO estimation accuracy is investigated.  
\end{abstract}



\maketitle

\section{Introduction}
\label{sec:Intro}
\IEEEPARstart{I}{n} digital communication systems, accurate synchronization is crucial for the correct reception and interpretation of signals. 
Importantly, the residual synchronization errors that remain after the initial alignment stage are especially challenging, as even small mismatches can result in significant degradation of system performance \cite{Kai_2004, Blumenstein_2018, Qi_2025, Morelli_2010, Son_2025, Li_2024, Wang_2025}.
Among such synchronization challenges, sampling frequency offset (SFO) is particularly important to mitigate in many applications and becomes more crucial in modern high-speed and wide-band communication systems, where the use of large signal bandwidth necessitates very high sampling rates, and thus even tiny differences between sampling clocks lead to a noticeable accumulative timing drift resulting in inter-carrier and inter-symbol interferences
\cite{Morelli_2010, Son_2025, Li_2024, Wang_2025, Matin_2021, Wang_2003, Giroto_2025, Brunner_2025}. 

The problem of SFO arises in several practical scenarios. In the point-to-point case, the transmitter and receiver operate on independent sampling clocks, and SFO must be estimated and corrected at the receiver to enable consecutive demodulation of the signal with high accuracy \cite{Morelli_2010, Son_2025, Matin_2021, Tsai_2005, Hou_2020, Parlin_2023, Oberli_2007, Qi_2023}. In distributed communication systems such as wireless acoustic sensor networks (WASN) \cite{Cherkassky_2017, Hu_2023} and bistatic integrated sensing and communication (ISAC) \cite{Giroto_2025, Brunner_2025}, this problem is escalated to multi-node synchronization as each remote node has an independent sampling clock, causing inter-node phase and timing drift, thereby degrading joint processing performance. Even within a single receiver, particularly, in massive multiple-input-multiple-output (mMIMO) systems, SFO can occur between groups of analog-to-digital converters (ADCs) under certain hardware architectures, leading to inconsistent timing across receiver chains and affecting further processing algorithms, e.g., reciprocity calibration, beamforming accuracy and localization \cite{Tian_2026}. 

Most of the existing methods for SFO estimation are designed for orthogonal frequency division multiplexing (OFDM) signals, and therefore carry out the estimation in the frequency domain \cite{Son_2025, Hou_2020,Li_2024, Nguyen_2009, Tsai_2005, Kim_2011, Xu2014, WU2024, Oberli_2007, Giroto_2025, Brunner_2025}. 
Moreover, in OFDM systems, SFO  and carrier frequency offset (CFO) are commonly estimated jointly, because they both appear as phase rotations per carrier 
or are coupled through a shared reference oscillator, creating frequency-dependent drifts \cite{Son_2025, Nguyen_2009, Tsai_2005, Kim_2011, Xu2014, WU2024, Oberli_2007}. 

Alternatively, a time-domain SFO estimation method utilizing the structure of OFDM signals has been proposed to reduce the implementation complexity \cite{Oswald_2004}, but the dependency on other errors has not been investigated in this study.
Additionally, a joint time-domain estimation of wireless channel responses, SFO, and CFO has been introduced for arbitrary-waveform signals \cite{Parlin_2023}; however, this approach relies on least-mean-squares (LMS) algorithms, which do not guarantee convergence without proper step-size parameters tuning \cite[Prop.~2, p. 508]{hayes1996}.
Furthermore, SFO estimation methods assume perfect alignment of the initial samples or rely on compensating the sampling time offset (STO) later in the receiver chain. Even though STO compensation is generally not required directly after sampling the signal, its presence still introduces an additional challenge for achieving accurate SFO estimation.

 Existing SFO estimation algorithms generally provide high estimation accuracy, but their computational complexity is often substantial, which leads to increased power consumption and reduced suitability for resource-constrained systems. Furthermore, many research works discuss implementation complexity only briefly and do not provide detailed complexity analyses, making it difficult to compare methods in terms of computational resources. It is also important to note that existing techniques typically work on a complex-valued signal, 
 consequently increasing the number of real multiplications\footnote{\label{foo:complex} One complex multiplication requires a minimum of three real multiplications \cite{Fam_1988}.} required for the estimation.

Digital compensation of SFO can be performed either in the frequency domain (commonly used for multi-carrier systems), by applying subcarrier-dependent phase rotations \cite{Harish_2017, Son_2025, Gao_2024}, or in the time domain, where the received sampled signal is interpolated \cite{Tsai_2005, Hamila_1998, Hu_2018, Gao_2024, Oberli_2007, Son_2025}.  
As discussed and analyzed in \cite{Son_2025} and \cite{Gao_2024}, with an increase of the SFO, time-domain digital interpolator-based methods significantly outperform the frequency-domain compensators, especially for wideband signals and systems with small subcarrier spacing, by providing more precise correction. Moreover, for general bandlimited waveforms that do not require FFT processing, time-domain SFO compensation is inherently more efficient, avoiding the additional computational overhead associated with frequency-domain operations.
Existing methods typically carry out the SFO compensation by interpolating the signal using the Farrow structure \cite{Hamila_1998, Hu_2018, Gao_2024, Erup_1993}. It is important to note that the Farrow structure is also utilized for a fixed (time-invariant) fractional time-delay estimation and compensation \cite{Dooley_1999, Olsson_2006, Tseng_2012}, however, to the best of the authors' knowledge, it has not been explored for the SFO estimation before. 

\subsection{Contributions}\label{sec:contributions}
In this work, we propose time-domain SFO estimation algorithms that employ the Farrow structure used for the SFO compensation. This is an extension of the work \cite{deiokshak2025} presented at a conference by the authors. The main contributions of the paper are as follows.
\begin{enumerate}
	\item The implementation complexity of the proposed estimation algorithms is reduced by exploiting the Farrow structure already used for compensation, thereby avoiding additional hardware or computational blocks. In addition, the proposed method utilizes only one real-valued component for complex signals for estimation, while the resulting estimates are applied to both signal components during compensation, yielding further significant overall complexity reductions (see also Footnote~\ref{foo:complex}).
	\item The proposed Farrow based estimators work for arbitrary bandlimited signals, and thus impose no additional constraints on the waveform structure, and there is no need for the FFT computation inherent in frequency-domain methods. 
	\item The proposed algorithms handle the presence of STO and, importantly, enables joint SFO and STO estimation even in the presence of other synchronization errors, like CFO and phase offset (PO). Consequently, after SFO and STO estimation and compensation, any available CFO and PO estimation and compensation methods can be applied \cite{Singh_2025}.  
	\item The proposed algorithms, which estimate two parameters (SFO and STO), rely on iterative updates computed using two complementary methods. To achieve rapid and precise refinement of the SFO estimate, we employ a Newton-based update rule. To further reduce implementation complexity, an alternative update algorithm based on an iterative least-squares (ILS) formulation is introduced. For both methods, an efficient way of computing time-index powered weighted sums using cascaded accumulators is used, that significantly reduces the implementation complexity by eliminating considerable parts of the multiplications \cite{deiokshak2026}. Thereby, the overall implementation complexity of the estimators scales only linearly with the number of samples employed. Moreover, to further reduce complexity while ensuring convergence, both update methods are simplified using a first-degree Farrow structure, leading to a common and computationally efficient solution.
	\item The proposed approach to SFO estimation targets mainly the problem of ADCs synchronization at a single receiver, where the output of one ADC is taken as the reference; however, the proposed method can also be used in other scenarios considering that the reference signal is known or partially known at the receiver. 
\end{enumerate}

\subsection{Outline}
Following this introduction, Section \ref{sec:compensation} gives a brief overview of the SFO problem, its compensation based on the Farrow structure, and states the proposed SFO estimation principle using the Farrow structure.
In Section \ref{sec:Newton}, the proposed estimation method using Newton's method is introduced, followed by the proposed ILS method in Section \ref{sec:ILS}, with simulation results showing the application and performance of the proposed estimators presented in Section \ref{sec:results}. In Section \ref{sec:one_branch}, simplified (first-degree) Farrow-based estimators are shown. Finally, Section \ref{sec:conclusions} concludes the paper.

\section{Problem Statement and Proposed SFO Time-Domain Estimation Utilizing the Farrow Structure}
\label{sec:compensation}
\subsection{SFO Model}
Let $x_a(t)$ represent a bandlimited continuous-time signal to be sampled with slightly different sampling frequencies, $f_0$ and $f_1 = f_0 + \Delta_{f}$, where $\Delta_{f}$ denotes the SFO between them. The sampled signals can be expressed as
\begin{align}
	x_0(n)=x_a(nT), \quad x_1(n)=x_a(n(1+\Delta)T+\varepsilon), 
	\label{eq:problem_offset}
\end{align}
where $T = 1/f_0$ is the sampling period of the reference signal,  $\Delta=-\Delta_{f}/(f_0+\Delta_{f})$ represents the difference between the sampling periods, and $\varepsilon$ is the STO between the two signals. It is important to note that while the SFO, $\Delta_{f}$, and thereby $\Delta$, remain constant over time, the time deviation, $n\Delta T$ in \eqref{eq:problem_offset}, between samples $x_0(n)$ and $x_1(n)$, accumulates with increasing sample index \(n\).
The challenge here is to accurately estimate $\Delta$ and $\varepsilon$ to adequately compensate for these discrepancies. 
Following the common assumption, we consider that a coarse synchronization stage has already been performed, and only a small residual error remains, such that
\begin{align}
	|n\Delta+\varepsilon|\leq d_{\text{max}}, \quad d_{\text{max}}=0.5
	\label{eq:max_d}
\end{align}
for $n=0,1,\dots,N-1$ with $N$ being the number of samples being used for estimation.\footnote{The choice $d_{\text{max}}=0.5$ is justified in Section~\ref{sec:comp_farrow_structure}.} Throughout the paper, estimating $\Delta$ will denote SFO estimation, as the SFO follows directly from $\Delta$. 

\subsection{SFO Compensation Using the Farrow Structure}
\label{sec:comp_farrow_structure}
\begin{figure}[tbp]
	\centering \includegraphics[scale=1.0]{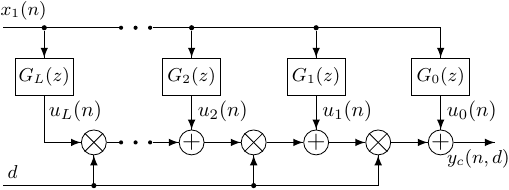}
	\caption{Variable-fractional-delay filter based on the Farrow structure.} 
	\label{Flo:farrow-equivalence-scheme}
\end{figure}
The Farrow structure \cite{farrow1998}, depicted in Fig. \ref{Flo:farrow-equivalence-scheme}, allows for adjusting a time-varying fractional delay $d$ in real time, ensuring high accuracy and efficient implementation due to the use of fixed-coefficient linear-phase finite-impulse-response (FIR) filters $G_k(z)$ \cite{Valimaki_1995, johansson2003}. In the context of the SFO considered here, the fractional delay $d$ is represented in terms of $\Delta$ and $\varepsilon$ as
\begin{equation}
	d(n, \Delta, \varepsilon)= n\Delta + \varepsilon,
	\label{eq:d}
\end{equation}
with the compensated signal\footnote{\label{foo:complex_value}For a complex-valued signal, each component (real and imaginary parts) must be compensated separately, i.e., using two instances of the Farrow structure. It is noted that such compensation is needed in all SFO correction techniques independently of the estimation method.} given by
\begin{equation}
y_c(n, \Delta, \varepsilon) = \sum_{k=0}^{L} d^k(n, \Delta, \varepsilon) u_k(n),
\label{eq:farrow_output}
\end{equation}
where $u_k(n) = x_1(n) * g_k(n)$ represents the output of subfilter $G_k(z)$, with $*$ denoting convolution. 
Further, throughout the paper, $d$ is used instead of $d(n, \Delta, \varepsilon)$ to simplify the expressions.

The  filter with a variable fractional delay (VFD) $d$ based on the Farrow structure can be designed by obtaining the subfilter impulse response coefficients $g_k(n)$ using any of the available methods, e.g., \cite{johansson2003, Deng_2006,  Huang_2018, Zhao_2019, Johansson_2013}, yielding $L$ optimized impulse responses of length $N_G+1$, with $G_0(z)$ being a pure delay when considering linear-phase FIR implementations, i.e., $G_0(z)=z^{-N_G/2}$.\footnote{While the subfilters can be of different orders \cite{johansson2003, Deng_2006}, here, we consider the same order for all subilters $G_k(z)$, $k=1,\dots, L$, for simplicity.}  
Then, the overall frequency response of the VFD filter is given by
\begin{align}
	G(e^{j\omega T}, d) = \sum_{k=0}^{L} d^k  G_k(e^{j\omega T}),
	\label{eq:vfd_freqresp}
\end{align}
and the filter complexity depends on the order $N_G$ and degree $L$, which are determined by the desired approximation error $\delta_c$  and the cut-off frequency $\omega_c T {\in} (0, \pi)$. In this paper, we consider VFD filters designed in the minimax sense and even order $N_G$ since it results in a smaller number of multiplications required \cite{johansson2003}. That is, the VFD filter specification is
\begin{align}
	\left|G_e(e^{j\omega T}, d)\right| \leq \delta_c
\end{align}
with 
\begin{align}
	G_e(e^{j\omega T}, d) = G(e^{j\omega T}, d) - G_\text{des}(e^{j\omega T}, d)
\end{align}
for $\omega T \in [0, \omega_c T]$ and $|d|\leq 0.5$, with $G(e^{j\omega T}, d)$ as in \eqref{eq:vfd_freqresp} and $G_\text{des}(e^{j\omega T}, d)=e^{-j\omega T(d+N_G/2)}$ being the desired response. Reducing the approximation error $\delta_c$ and increasing $\omega_c T {<} \pi$ tighten the requirements for the filter and result in higher order $N_G$ and a greater number of branches $L+1$. 

\subsection{Proposed SFO Estimation Principle Using the Farrow Structure} 
\label{sec:estimator}
In this paper, to estimate the SFO, we propose incorporating STO as the second parameter for a joint estimation and employ the compensator described above, specifically, the output of the Farrow structure $y_c(n, \Delta, \varepsilon)$ in \eqref{eq:farrow_output}, which is assumed to be a real-valued signal (see Item 1 in Section \ref{sec:contributions} and Footnote~\ref{foo:complex_value}). The objective function is expressed as the scaled-by-half\footnote{\label{foo:fraction}The factor $1/2$ is used in \eqref{eq:cost} for notation simplicity; it cancels out when taking the derivatives of the cost function.} squared error (SE) between the compensated signal $y_c(n, \Delta, \varepsilon)$ and the reference signal $x_0(n)$ as 
\begin{align}
    F(\Delta, \varepsilon) &= \frac{1}{2} \sum_{n=0}^{N-1} \Big(y_c(n, \Delta, \varepsilon) - x_0(n) \Big)^2.
	\label{eq:cost}
\end{align}
The objective is to estimate the parameters ${\Delta}$ and ${\varepsilon}$ by minimizing $F(\Delta, \varepsilon)$, thus solving the optimization problem
\begin{equation}
	\widehat{\Delta}, \widehat{\varepsilon} = \arg\min_{\Delta, \varepsilon} F(\Delta, \varepsilon),
	\label{eq:argmin_cost}
\end{equation}
which is nonlinear in parameters ${\Delta}$ and ${\varepsilon}$ for $L>1$. 
In the proposed approach, we solve this optimization problem iteratively according to the update rule

\begin{equation}
	\mathbf{w}^{(m+1)} = \mathbf{w}^{(m)} - \boldsymbol{\delta}_\mathbf{w}^{(m)},
	\label{eq:update}
\end{equation}
where a new set of parameters $\mathbf{w}^{(m+1)}=[\Delta^{(m+1)}, \varepsilon^{(m+1)}]^\top$ for iteration $m+1$ can be computed based on the values $\mathbf{w}^{(m)}$ at the previous iteration $m$ and the updates $\boldsymbol{\delta}_{\mathbf{w}}^{(m)}$, given by
\begin{align}
	{\boldsymbol{\delta}}_\mathbf{w}^{(m)} = 
	\begin{bmatrix}
		{\delta}_\Delta^{(m)}, & {\delta}_\varepsilon^{(m)}
	\end{bmatrix}^\top.
	\label{eq:update_delta}
\end{align}
These updates are computed according to two complementary methods, which will be discussed in Sections \ref{sec:Newton} and \ref{sec:ILS}, respectively.
Throughout the paper, we assume the starting point for the iterative optimization problem in \eqref{eq:update} being
\begin{align}
	\Delta^{(0)}=\varepsilon^{(0)}=0.
	\label{eq:starting_point}
\end{align}

\section{Proposed SFO Estimator Using Newton's Method}
\label{sec:Newton}
Given the objective function in \eqref{eq:cost}, which is twice differentiable, we propose to use Newton's method \cite[Sec.~9.5, pp.~484--496]{boyd2004convex} (a.k.a. Newton-Raphson) to compute the updates in \eqref{eq:update_delta}. This method uses second-order information from the Hessian matrix, allowing it to account for the local curvature of the cost function and thereby take steps with rapid convergence
\cite{kochenderfer2019algorithms}.
It is implemented through the update rule in \eqref{eq:update}, with
\begin{align}
	\boldsymbol{\delta}_{\textbf{w}}^{(m)}=\big(\mathbf{H}^{(m)}\big)^{-1} \mathbf{g}^{(m)}
	\label{eq:delta_w_Newton}
\end{align}
being the Newton's update for the parameters $\mathbf{w} = [\Delta, \varepsilon]^{\top}$. Here, $\mathbf{g}=\nabla F(\Delta, \varepsilon)$ stands for the gradient of the cost function defined in \eqref{eq:cost}, $\mathbf{H}=\nabla^2 F(\Delta, \varepsilon)$ for its Hessian matrix, and $m$ for the iteration index. To facilitate the computation of the updates, both the gradient and Hessian matrix are expressed as functions of $d$ according to \big(recall that $\frac{\partial F}{\partial \Delta} =\frac{\partial F}{\partial d} \frac{\partial d}{\partial \Delta}$ and $\frac{\partial F}{\partial \varepsilon} =\frac{\partial F}{\partial d} \frac{\partial d}{\partial \varepsilon}$\big)
\begin{align}
	\hspace{-0.15cm}
   \mathbf{H}= \begin{bmatrix}
        \displaystyle \frac{\partial^2 F}{\partial \Delta^2} & \displaystyle \frac{\partial^2 F}{\partial \Delta \partial\varepsilon} \\ \\
        \displaystyle  \frac{\partial^2 F}{\partial\varepsilon \partial \Delta} & \displaystyle\frac{\partial^2 F}{\partial \varepsilon^2}
    \end{bmatrix}
    & =\begin{bmatrix}
    \displaystyle\sum_{n=0}^{N-1} n^2 \frac{\partial^2 F_n}{\partial d^2} & \displaystyle\sum_{n=0}^{N-1} n \frac{\partial^2 F_n}{\partial d^2} \\ \\
    \displaystyle\sum_{n=0}^{N-1}  n \frac{\partial^2 F_n}{\partial d^2} & \displaystyle\sum_{n=0}^{N-1}  \frac{\partial^2 F_n}{\partial d^2}
    \end{bmatrix}\hspace{-0.12cm},\label{eq:Hessian} \\
	\mathbf{g} =  
    \begin{bmatrix} 
        \displaystyle\frac{\partial F}{\partial \Delta} &
        \displaystyle\frac{\partial F}{\partial \varepsilon}
    \end{bmatrix}^{\top}
	&=\begin{bmatrix} 
		\displaystyle\sum_{n=0}^{N-1}n\frac{\partial F_n}{\partial d} &
		\displaystyle\sum_{n=0}^{N-1}\frac{\partial F_n}{\partial d} 
	\end{bmatrix}^{\top}\hspace{-0.1cm}, \label{eq:gradient}
\end{align}
where  $\displaystyle F=F(\Delta, \varepsilon)=\sum_{n=0}^{N-1} F_n$ as in \eqref{eq:cost}, with
\begin{align}
	F_n=\frac{1}{2}\big(y_c(n, \Delta, \varepsilon) - x_0(n)\big)^2.
\end{align}  
Here, the first- and second-order derivatives of $F_n$ are
\begin{align}
	\frac{\partial F_n}{\partial d} 
	&=  \Bigg(\sum_{k=0}^{L} d^k u_k(n)-x_0(n)\Bigg) \Bigg(\sum_{k=1}^{L} k d^{k-1} u_k(n)\Bigg), \hspace{-3.0 pt} \label{eq:first_derivative}\\
	\frac{\partial^2 F_n}{\partial d^2} 
	&= \Bigg(\sum_{k=1}^{L} k d^{k-1} u_k(n)\Bigg)^2 + \Bigg(\sum_{k=0}^{L} d^k u_k(n)-x_0(n)\Bigg) \nonumber \\ 
	&\times\Bigg(\sum_{k=2}^{L} k(k-1) d^{k-2} u_k(n)\Bigg). \label{eq:second_deriv}
\end{align}

Provided that the fractional delay at the first iteration (i.e., $m{=}1$) satisfies $|d|{\le}0.5$ according to \eqref{eq:max_d} and \eqref{eq:d}, the algorithm operates within the intended fractional-delay region of the VFD filter; moreover, since the cost function in \eqref{eq:cost} is a polynomial in $d$, it is twice differentiable with a locally Lipschitz continuous Hessian. If, in addition, the Hessian is positive definite in a neighborhood of a stationary point,\footnote{\label{foo:Hessian_pd}In all numerical experiments conducted, the Hessian evaluated at every iteration remained positive definite, with no indications of numerical ill-conditioning. 
However, we have not been able to prove it mathematically, except for the special case of $L=1$. This will be shown in Section~\ref{sec:ILS} for the ILS, where the corresponding matrix $\mathbf{Q}$ in \eqref{eq:Q} coincides with the Hessian in \eqref{eq:Hessian}.} then Newton’s method exhibits local quadratic convergence \cite[Th.~3.5, p.~44]{nocedal2006numerical}. 

\begin{remark}\label{rem:starting_point}
At the first iteration, with the initialization in \eqref{eq:starting_point}, the derivatives in \eqref{eq:first_derivative} and \eqref{eq:second_deriv} simplify respectively to
\begin{align}
		\left.\frac{\partial F_n}{\partial d}\right|_{d=0} &= \Big(u_0(n)-x_0(n)\Big) u_1(n),\label{eq:remark1_first}\\
		\left.\frac{\partial^2 F_n}{\partial d^2}\right|_{d=0} & =\Big(u_1(n)\Big)^2 + 2\Big(u_0(n)-x_0(n)\Big) u_2(n).\label{eq:remark1_second}
\end{align}
This means that the terms for $k\geq3$ are not utilized in the first iteration of the estimator, i.e., the VFD filter in \eqref{eq:vfd_freqresp} simplifies to the second-degree Farrow structure.
\end{remark}

\begin{remark}\label{rem:L1}
	For the first-degree Farrow structure, i.e., for the case when $L=1$, the expressions in \eqref{eq:first_derivative} and \eqref{eq:second_deriv} are simplified respectively to 
	\begin{align}
		\left.\frac{\partial F_n}{\partial d}\right|_{L=1} &=\big(u_0(n)+d u_1(n)-x_0(n)\big) u_1(n),\\
		\left.\frac{\partial^2 F_n}{\partial d^2}\right|_{L=1} & = \big(u_1(n)\big)^2,
	\end{align}
	leading to the linear-least-squares solution as will be discussed in Section~\ref{sec:one_branch}.
\end{remark}

\subsection{Low-Complexity Implementation} \label{sec:complexity}
The proposed Newton's based estimator can be implemented cost-efficiently using the computations of the derivatives presented in Fig. \ref{Flo:derivatices}. Furthermore, we introduce an efficient method below for calculating the elements in \eqref{eq:Hessian} and \eqref{eq:gradient}.

\begin{figure*}[tbp]
	\centering
	\centering \includegraphics[scale=0.93]{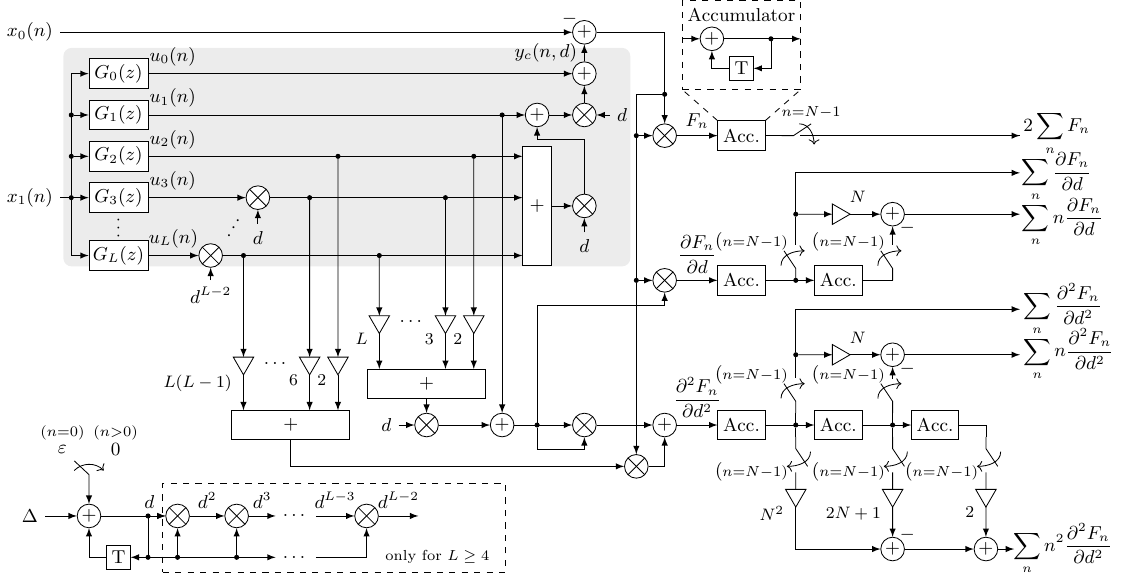}
	\caption{Implementation of the SFO compensator (highlighted with the gray rectangular) and computations of the elements in \eqref{eq:Hessian} and \eqref{eq:gradient} required for the SFO estimator using the Newton based algorithm.}\label{Flo:derivatices}
\end{figure*}
\subsubsection{Time-Index and Time-Index-Squared Weighted Sum Computations Based on Cascaded Accumulators} \label{sec:accumulators}
To efficiently calculate the time-indexed weighted (TIW) sums, $\sum_n n \frac{\partial F_n}{\partial d}$ and $\sum_n n \frac{\partial^2 F_n}{\partial d^2}$, and time-index-squared weighted (TISW) sum, $\sum_n n^2 \frac{\partial^2 F_n}{\partial d^2}$, in \eqref{eq:Hessian}--\eqref{eq:gradient}, related to the first and second derivatives, we propose an efficient implementation based on cascaded accumulators \cite{deiokshak2025}. Denoting the input to the first accumulator $v_n= \frac{\partial^2 F_n}{\partial d^2}$, the second derivative in \eqref{eq:Hessian}, it can be shown that the output of the first, second, and third cascaded accumulators can be expressed in terms of the input $v_n$ as\footnote{A detailed proof is provided in \cite{deiokshak2026}.}
\begin{align}
	A_1(N) &= \sum_{n=0}^{N-1} v_n, \\
	A_2(N) &= \sum_{n=0}^{N-1} (N-n)v_n,  \label{eq:accum2}\\
	A_3(N) &= \sum_{n=0}^{N-1} \frac{(N-n)(N-n+1)}{2} v_n.
\end{align}
To generate the matrix elements in \eqref{eq:Hessian}, the output of the first accumulator, $A_1(N)$, directly gives the sum of $v_n$, while the TIW and TISW sums are computed by utilizing the outputs of the accumulators scaled by constants as
\begin{align}
	\sum_{n=0}^{N-1} n v_n   &= NA_1(N) - A_2(N),  \label{eq:S2} \\ 
	\sum_{n=0}^{N-1} n^2 v_n & = N^2 A_1(N) - (2N+1)A_2(N) +2A_3(N). \label{eq:S3}
\end{align} 
Likewise, the elements in \eqref{eq:gradient} can be obtained using two cascaded accumulators, with the input to the first accumulator set as $v_n = \frac{\partial F_n}{\partial d}$ and with the outputs computed as in \eqref{eq:accum2}. Then the desired TIW sum can be calculated as in \eqref{eq:S2}.

\subsubsection{Implementation Complexity}
The implementation complexity of the proposed estimator comprises the update computations as described in \eqref{eq:update}, with the subsequent matrix inversion in \eqref{eq:delta_w_Newton}, and the computation of the derivatives to obtain the elements in \eqref{eq:Hessian} and \eqref{eq:gradient}, via the proposed method for the sum of multiplications using cascaded accumulators outlined above. 

To assess the complexity, we distinguish between general multiplications and fixed multiplications since the latter can be efficiently implemented using adders and shift operators \cite{Sarband_2025}. Additionally, some parts of the estimation algorithm, specifically the subfilter outputs $u_k(n)$, $k = 0, \dots, L$, multiplications by $d$, and additions required to compute the output $y_c(n,d)$, are utilized for compensation (in Fig.~\ref{Flo:derivatices}, this part is highlighted with a gray rectangle). Given that they are needed regardless of the estimation method, their cost is excluded from the overall complexity of the estimator. 
Further, the total complexity of the Newton's based SFO estimator is constructed from the following parts.
\begin{itemize}
	\item Computing $d$ and its powers based on $\Delta$ and $\varepsilon$ (up to $d^{L-2}$ as can be seen at the bottom of Fig.~\ref{Flo:derivatices}) requires $\max\{L{-}3, 0\}$ general multiplications and one addition per sample. 
	\item Obtaining the elements $\partial F_n/\partial d$ and $\partial^2 F_n/\partial d^2$ requires four general multiplications, $2L-2$ fixed multiplications, and $2L{-}1$ additions per sample for $L\geq2$, and no fixed multiplications, two general multiplications and one addition per sample for $L{=}1$.
	\item Each accumulator in \eqref{eq:S2}--\eqref{eq:S3} needs $N-1$ additions per $N$ samples, and five fixed multiplications and four additions per $N$ samples are required to obtain the elements in \eqref{eq:Hessian} and \eqref{eq:gradient} based on the accumulator outputs.
	\item The multiplication of the $2\times 2$ matrix inverse and $2\times 1$ vector in \eqref{eq:delta_w_Newton}, where $h_a,h_b,h_c$ denote the scalar entries of the Hessian matrix $\mathbf H$ and $g_a,g_b$ the corresponding elements of the gradient vector $\mathbf g$, can be computed as
	\begin{align}
		\mathbf{H}^{-1}\mathbf{g}=&
		\begin{bmatrix}
			h_{a} & h_{b}\\
			h_{b} & h_{c}
		\end{bmatrix}^{-1} 
		\begin{bmatrix}
			g_{a}\\
			g_{b}
		\end{bmatrix}\notag\\
		=&
		\frac{1}{h_{a}h_{c}-h_{b}^2}
		\begin{bmatrix}
			h_{c}g_{a} -h_{b}g_{b}\\
			-h_{b}g_{a} + h_{a}g_{b}
		\end{bmatrix},
	\end{align}
	which requires eight general multiplications, three additions, and one division.
	\item Update in \eqref{eq:update} needs two additions per batch of $N$ samples.
\end{itemize}
Considering all the implementation aspects mentioned above, the proposed estimation algorithm requires $\max\{L{+}1, 4\}{\times} N {+} 8$ general multiplications, $(2L{-}2){\times} N{+}5$ fixed multiplications, $(2L{+}5){\times} N{+}4$ additions, and one division per batch of $N$ samples for $L\geq2$, and $2{\times}N {+} 8$ general multiplications, five fixed multiplications, $7{\times} N{+}4$ additions, and one division per batch of $N$ samples for $L{=}1$. 

These expressions provide an upper bound for the implementation complexity because in special cases the structure in Fig.~\ref{Flo:derivatices} can be simplified and some operations can be eliminated considering the starting point for the estimation [see Remark~\ref{rem:starting_point} above \eqref{eq:remark1_first}]. 
Note that the complexity of obtaining the output $2\sum_{n=0}^{N-1} F_n$ (i.e., one general multiplication and one accumulator) is not taken into account for the overall estimator complexity because the cost function does not need to be computed in practical implementations if the number of iterations $m$ is fixed or the stop criteria for the estimation problem is defined without accounting for the cost function value.
Further, in regular Farrow structures, $L$ is a small number, typically between one and five as will be seen in Section~\ref{sec:complexity_estimation_accuracy}.

It is important to highlight that the proposed estimator, which relies on the Newton's method, requires only a few iterations to converge if the number of samples used for estimation is sufficient relative to the noise level, which will be discussed in the next paragraph, and thus the total implementation complexity per sample is low. Moreover, unlike frequency-domain estimation methods, the proposed algorithm does not require the FFT computation and uses only real-valued signals, thus the number of arithmetic operations is reduced. This results in linear implementation complexity of the proposed Newton based estimator, and will be further discussed in Section~\ref{sec:ILS_complexity}.

\subsubsection{Number of Samples Used for Estimation}
\label{sec:number_samples}
In general, VFD filters are designed to handle only the fractional component of the delay between adjacent samples. As a result, limiting the delay parameter to the interval $d\in[-0.5, 0.5]$ is sufficient.\footnote{In general, the interval for $d$ can be increased, but the implementation complexity of the VFD filter will be significantly higher due to a greater number of branches and higher filter orders \cite{Tseng_2012}.} Thus, in the case of estimation using the Farrow filter, with $d$ as defined in \eqref{eq:d}, we get the bound in \eqref{eq:max_d}.
This gives the following upper bound of the number of samples $N$ per iteration used for the estimation methods based on the Farrow structure as $N \leq (0.5 - \varepsilon)/\Delta$.
Further, to ensure the stability of the proposed estimation methods and to guarantee that the matrix $\mathbf{H}$ in \eqref{eq:update} is full rank (see Footnote~\ref{foo:Hessian_pd}), the number of samples must exceed the number of estimated parameters by at least one. Consequently, a lower bound of $N>2$ is required. This implies that 
$N$ must be selected within the interval
\begin{align}
	2 < N \leq \frac{0.5 - \varepsilon}{\Delta}.
	\label{eq:N_bounds}
\end{align}
In the presence of noise, the number of samples must be significantly larger than the number of parameters  to ensure accurate parameter estimation. At the same time, 
$N$ also directly affects the computational and implementation complexity of the estimation methods as discussed above. Therefore,  $N$ should be chosen as small as possible within the bounds in \eqref{eq:N_bounds} while still providing satisfactory estimation performance, resulting in a trade-off between estimation accuracy and complexity as will be shown in Section~\ref{sec:ex_number_of_samples}.

 \begin{figure*}[t!]
	\centering
	\centering \includegraphics[scale=0.93]{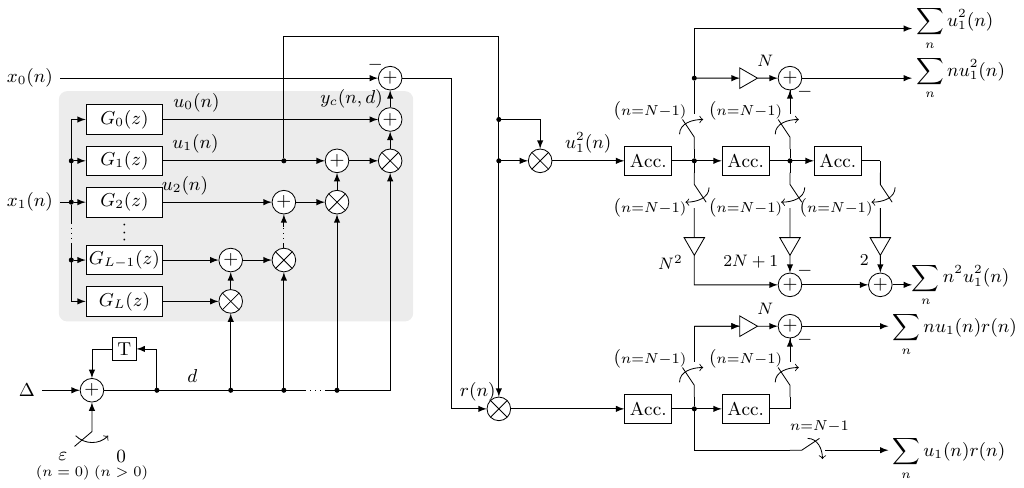}
	\caption{Implementation of the SFO compensator (highlighted with the gray rectangular) and computations to obtain the elements in \eqref{eq:Q} and \eqref{eq:c} required for the SFO estimator using the ILS algorithm.}\label{Flo:ILS_implem}
\end{figure*}

\section{Proposed SFO Estimator Using An Iterative Least-Squares Algorithm} \label{sec:ILS}
The estimation problem in \eqref{eq:argmin_cost} with the objective function in \eqref{eq:cost} involves the compensated signal $y_c(n,\Delta,\varepsilon)$ generated by the VFD filter. Here, the nonlinearity in the parameters $\Delta$ and $\varepsilon$ arises from the higher-order terms indexed by $k \geq 2$. 
As mentioned in Section~\ref{sec:comp_farrow_structure}, in typical fractional-delay applications, the delay parameter $d$ is bounded as $|d|{\leq}0.5$, which is the same as the assumption in \eqref{eq:max_d}, so that $(n\Delta+\varepsilon)^k$ in \eqref{eq:farrow_output} decreases in magnitude as $k$ increases. 

At the first iteration with the starting point in \eqref{eq:starting_point}, i.e., $d=0$, it is noted from Remark~\ref{rem:starting_point} that the subfilter outputs $u_k(n)$ for $k\geq3$ are not utilized for estimation using the Newton method. Further, at iterations $m\geq1$ for iterative-based optimization algorithms as in \eqref{eq:update} and with the assumption in \eqref{eq:max_d}, the updates ${\boldsymbol{\delta}}^{(m)}=[\delta_\Delta^{(m)}, \delta_\varepsilon^{(m)}]^\top$ are typically smaller than the true parameters $[\Delta_\text{true}, \varepsilon_\text{true}]$ provided that the starting point is close to the true solution and the gradient in \eqref{eq:gradient} is close to zero. This means that the difference between $(n(\Delta^{(m)}+\delta_\Delta^{(m)})+\varepsilon^{(m)}+\delta_\varepsilon^{(m)})^k$ and $(n\Delta^{(m)}+\varepsilon^{(m)})^k$ is very small for $k\geq2$. 

Motivated by the above mentioned aspects, we propose to linearize the nonlinear problem in \eqref{eq:argmin_cost} by using a first-order approximation model for computing the updates ${\boldsymbol{\delta}}^{(m)}$, i.e., assuming that only linear terms of the Farrow structure ($L=1$) are used for the estimation updates, while all the terms are used for compensation, i.e. computing the outputs $y_c(n,\Delta, \varepsilon)$  in \eqref{eq:farrow_output}. This approach allows to reduce the implementation complexity of the estimation comparing to the Newton based algorithm discussed in Section~\ref{sec:Newton} at the price of slightly less accurate steps at each iteration compare to the full-Farrow structure used for the estimation as will be shown in Section~\ref{sec:results}. 
Thus, at iteration $m$, let $(\Delta^{(m)},\varepsilon^{(m)})$ be the current parameter estimates and define the compensated output $y_c(n, \Delta^{(m)}, \varepsilon^{(m)})$ as in \eqref{eq:farrow_output}. Then, the corresponding residual is given by
\begin{align}
	r^{(m)} (n)&= y_c(n, \Delta^{(m)},\varepsilon^{(m)}) - x_0(n).
	\label{eq:residual_m}
\end{align}
 
Using the update rule in \eqref{eq:update}, the compensator output $y_c(n, \Delta^{(m+1)},\varepsilon^{(m+1)})$ at iteration $(m+1)$ can be expressed using the parameters $(\Delta^{(m)},\varepsilon^{(m)})$ from the previous iteration $m$, and thus the 
residual at iteration $(m+1)$ is
\begin{align}
	r^{(m+1)}(n) =&\sum_{k=0}^{L}\Big((n\Delta^{(m)}+\varepsilon^{(m)})-(n\delta_\Delta^{(m)}+\delta_\varepsilon^{(m)})\Big)^k\notag\\
	&\times u_k(n)-x_0(n) \notag\\
	=
	& \underbrace{\sum_{k=0}^{L}\Big((n\Delta^{(m)}+\varepsilon^{(m)})\Big)^k u_k(n)-x_0(n)}_{r^{(m)} (n)}\notag\\
	&-\Big((n\delta_\Delta^{(m)}+\delta_\varepsilon^{(m)})u_1(n) + e(n, \delta_{\Delta}^{(m)}, \delta_{\varepsilon}^{(m)})\Big),
	\label{eq:residual_m1}
\end{align}
where $e(n, \delta_{\Delta}^{(m)}, \delta_{\varepsilon}^{(m)})$ represents all higher order terms related to $\delta_\Delta^{(m)}$ and $\delta_\varepsilon^{(m)}$.
Assuming that their contribution to the overall residual error $r^{(m+1)}(n)$ is very small, \eqref{eq:residual_m1} can be approximated by
 \begin{align}
 	\tilde{r}^{(m+1)} (n)
 	&\approx r^{(m)} (n) - (n\delta_\Delta^{(m)}+\delta_\varepsilon^{(m)})u_1(n).
 	\label{eq:ILS_apprpox}
 \end{align}
 To obtain the parameters $\Delta$ and $\varepsilon$ using the  iterative update rule in \eqref{eq:update}, the residuals $\tilde{r}^{(m+1)}(n)$ must be minimized over $n=0,\dots, N-1$ at each iteration $m$. Thus, the updates $\boldsymbol{\delta}_\textbf{w}^{(m)}=[\delta_\Delta^{(m)},\delta_\varepsilon^{(m)}]^\top$ at iteration $m$ can be obtained in the least-squares sense
 \begin{align}
 	\min_{\delta_\Delta^{(m)},\delta_\varepsilon^{(m)}}
 	\sum_{n} \left(r^{(m)}(n) - (n\delta_\Delta^{(m)} + \delta_\varepsilon^{(m)})u_1(n) \right)^2. 
 	\label{eq:ILS_problem}
 \end{align}
  In matrix form, this can be written as 
 \begin{align}
 	\min_{\boldsymbol{\delta}_{\mathbf{w}}^{(m)}} 
 	\left\lVert \mathbf{r}^{(m)} - \mathbf{A}\,\boldsymbol{\delta}_{\mathbf{w}}^{(m)} \right\rVert^{2},
 \end{align}
 with
 \begin{align}
 	\mathbf{A}
 	&{=} 
 	\begin{bmatrix}
 		0 & u_1(0) \\
 		u_1(1) & u_1(1) \\
 		2u_1(2) & u_1(2) \\
 		\vdots       & \vdots   \\
 		(N-1) u_1(N-1) & u_1(N-1)
 	\end{bmatrix},	\label{eq:ILS_A}\\
 	\mathbf{r}^{(m)}
 	&{=}
 	\begin{bmatrix}
 		r^{(m)}(0) \,\,\,
 		r^{(m)}(1) \,\,\,
 		{\dots} \,\,\,
 		r^{(m)}(N-1)
 	\end{bmatrix}^{\top} \hspace{-4 pt},
	\label{eq:ILS_r}
 \end{align}
 and $\boldsymbol{\delta}_\textbf{w}^{(m)}$ as in \eqref{eq:update_delta}.
 
 Thus, the corresponding least-squares estimator gives the updates
 \begin{align}
 	\boldsymbol{\delta}_{\textbf{w}}^{(m)}
 	&=\Big(\textbf{A}^\top\textbf{A}\Big)^{-1}\textbf{A}^\top\mathbf{r}^{(m)}\notag\\
 		&={\textbf{Q}}^{-1} {\textbf{c}}^{(m)}\label{eq:ILS_update}
 \end{align}
 with 
 \begin{align}
 	{\textbf{Q}} &= 
 	\begin{bmatrix}
 		\displaystyle\sum_{n=0}^{N-1} n^2 u_1^2(n) & \displaystyle\sum_{n=0}^{N-1} n u_1^2(n)\\
 		\displaystyle\sum_{n=0}^{N-1} n u_1^2(n) &\displaystyle \sum_{n=0}^{N-1} u_1^2(n)
 	\end{bmatrix},\label{eq:Q}\\
 	\textbf{c}^{(m)} &= 
 	\begin{bmatrix}
 		\displaystyle\sum_{n=0}^{N-1} n u_1(n) r^{(m)}(n) & \displaystyle\sum_{n=0}^{N-1} u_1(n) r^{(m)}(n)
 	\end{bmatrix}^\top,\label{eq:c}
 \end{align}
given than $\mathbf Q{=}\mathbf A^\top\mathbf A$ is invertible, wich is formalized in the following theorem,
\begin{theorem} \label{teo:invertibility}
	$\mathbf{Q}=\mathbf{A}^\top \mathbf{A}$ is invertible if there exist 
	at least two sample indices $n_a \neq n_b$ such that $u_1(n_a) \neq 0$ and 
	$u_1(n_b) \neq 0$ for any integer $n_a, n_b\in[0, N-1]$.
\end{theorem}
The proof of Theorem~1 is given in Appendix~\ref{app:ILS_invertibility}. It is important to note here that $\mathbf{Q}$ in \eqref{eq:Q} coincides  with $\mathbf{H}$ in \eqref{eq:Hessian} for $L{=}1$, as will be shown in Section~\ref{sec:one_branch} (see also Footnote~\ref{foo:Hessian_pd}).

\subsection{Low-Complexity Implementation}\label{sec:ILS_complexity}
\begin{table}[t!]
	\centering
	\renewcommand{\arraystretch}{1.2}
	\caption{Implementation Complexity of the Proposed SFO Estimation Algorithms per Block of $N$ Samples}
	\begin{tabular}{p{0.8cm}|P{1.45cm}|P{2.3cm}|P{1.6cm}|P{0.35cm}}
		\hline
		Method & Fixed mult. & General mult. &Additions & Div.\\
		\hline
		NM & \multirow{2}{*}{$(2L{-}2)N{+}5$}&\multirow{2}{*}{$\max\{L{+}1, 4\} N {+} 8$}&\multirow{2}{*}{$(2L{+}5) N{+}4$} & 1\\
		\vspace{-0.35cm}
		($L{\geq}2$) & & & \\
		NM & \multirow{2}{*}{$5$}&\multirow{2}{*}{$2N {+} 8$}&\multirow{2}{*}{$7N{+}4$} & 1\\
		\vspace{-0.35cm}
		($L{=}1$) & & & \\
		ILS & $5$ &$2N{+}8$&$7N{+}4$ & 1\\
		\hline
	\end{tabular}
	\label{tab:complexity}
	\vspace{0.2cm}
\end{table}
Similar to the Newton based algorithm discussed in Section~\ref{sec:Newton}, the implementation complexity of the ILS estimator consists of computing the elements of the matrices in \eqref{eq:Q} and \eqref{eq:c}, with the subsequent matrix inversion and the updates in \eqref{eq:update}. Thus, the implementation shown in Fig.~\ref{Flo:ILS_implem} requires $2{\times} N$ general multiplications (with two multiplications per sample), five fixed multiplications, and $2N {+} 5{\times}(N{-}1) {+}4$ additions (with $(N{-}1)$ additions required for each accumulator) for a block of $N$ samples. Additionally, eight general multiplications, three additions, and one division are required for the update rule in \eqref{eq:ILS_update}, and two additions are needed for the update in \eqref{eq:update}. The overall implementation complexities of the Newton based and ILS methods are summarized in Table~\ref{tab:complexity}. It follows that the complexity of the Newton based method depends on the degree $L$ of the VFD filter, while the ILS method complexity is the same regardless of the number of branches of the filter. Moreover, for the case when $L{=}1$, the complexities of both methods coincide. This special case will be discussed in Section~\ref{sec:one_branch}. 
Further, as for the Newton's based method, the number of samples required for estimation using the ILS method must be chosen within the bound in \eqref{eq:N_bounds}, and its influence on the estimation accuracy will be shown in Section~\ref{sec:ex_number_of_samples}.

It is worth noting that the matrix $\textbf{Q}$ in \eqref{eq:Q} does not depend on the parameters $\Delta$ and $\varepsilon$, i.e., it does not require to be recomputed after one iteration. This is different from the Newton based algorithm, where the elements of $\textbf{H}$ in \eqref{eq:Hessian} must be recomputed at each iteration based on $d$-value and thereby based on the updates $\boldsymbol{\delta}_\textbf{w}$.\footnote{Alternative Newton-type methods, such as the Gauss–Newton method \cite[Sec.~10.3, pp.~254-257]{nocedal2006numerical}, approximate the Hessian matrix and may reduce computational complexity. However, such variants are beyond the scope of this paper.} Therefore, if the estimation requires more than one iteration, the complexity savings of the ILS algorithm is even more compared to the Newton based method.

Compared to existing SFO estimation techniques, both of the proposed algorithms use only real-valued operations and their implementation complexity grows linearly with $N$. In contrast, the existing approaches typically process complex-valued signals and rely on specific waveform structures, that together limit applications and require greater number of arithmetic operations. 

Moreover, if frequency-domain SFO estimation methods are employed and the fast Fourier transform is required solely for estimation, the resulting implementation complexity scales as $N\log_2(N)$ complex multiplications.

\section{Numerical Examples} \label{sec:results}
\subsection{Accuracy of Proposed Estimators}
In this subsection, we evaluate the performance of the proposed methods by applying them to different types of band-limited signals. The objective is to synchronize two signals sampled at $f_1$ and $f_0$, where $f_0$ serves as the reference with a normalized sampling frequency of one. 

\example\label{ex:multisine} \textbf{Importance of jointly estimating SFO and STO}.
To illustrate the impact of the STO parameter $\varepsilon$ on the accuracy of the SFO estimation, we consider a real-valued multisine signal whose amplitudes and phases are selected from random 16-QAM symbols. The system is initialized with a relatively large offset $\varepsilon=-0.2$ and $\Delta=400$~ppm, corresponding to $|d(n)|\le 0.21$ when using $N=1024$ samples for the estimation since $n\in[0,N-1]$.
The signal-to-noise ratio (SNR) is set to $30$~dB. This configuration is deliberately chosen to emphasize the importance of jointly estimating $\Delta$ and $\varepsilon$.

When the estimators are applied with one iteration while neglecting the STO parameter, i.e., estimating only $\Delta$, the obtained estimates are $\widehat{\Delta}=101$~ppm for the ILS approach and $\widehat{\Delta}=103$~ppm for Newton’s method, resulting in an NMSE of approximately $2.7\times 10^{-2}$ for both estimators. This pronounced underestimation occurs because the STO-induced delay is not modelled; consequently, its effect is absorbed into the SFO parameter, leading to a strongly biased estimate of $\Delta$. Additional iterations provide only marginal refinement (below $3$~ppm in total) and therefore do not compensate for the fundamental bias introduced by neglecting $\varepsilon$.

In contrast, when both parameters are jointly estimated, the estimators yield $\widehat{\Delta}=400$~ppm and $\widehat{\varepsilon}=-0.201$~ppm for the ILS approach, and $\widehat{\Delta}=401$~ppm and $\widehat{\varepsilon}=-0.201$~ppm for Newton’s method, with an NMSE of $1.9{\times} 10^{-3}$ for both methods already after a single iteration.

Although a residual STO could in principle be mitigated by a subsequent equalization stage, an erroneous SFO estimate cannot be corrected afterward. Hence, accurate estimation of $\Delta$ fundamentally requires including $\varepsilon$ in the joint estimation, especially when $\varepsilon \gg \Delta$.  In the remainder of the examples, both parameters are therefore estimated simultaneously.

Figure~\ref{Flo:multisine} shows the reference signal $x_0(n)$, the SFO-affected signal $x_1(n)$, and the compensated signals obtained with the ILS method $y_{\text{ILS}}(n)$ and Newton’s method $y_{\text{NM}}(n)$. In both cases, the compensated signals closely align with the reference, confirming the necessity and effectiveness of the joint estimation approach.

\begin{figure}[tbp]
	\centering
	\centering \includegraphics[scale=1.0]{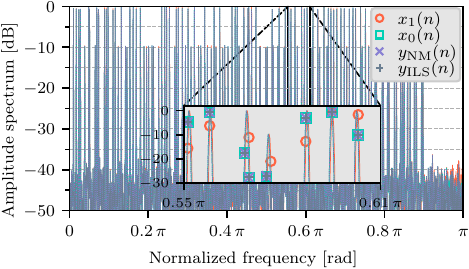}
	\caption{\textbf{\textit{Example~\ref{ex:multisine}:} Importance of jointly estimating (SFO and STO)}. Amplitude spectrum before and after compensation for a multisine signal.}
	\label{Flo:multisine}
\end{figure}

\example \label{ex:mcperformance} \textbf{Monte Carlo and convergence analyses.} To complement the representative single-realization example above, a statistical performance evaluation is conducted. Table~\ref{tab:performance} reports the averaged NMSE
\begin{align}
	\overline{\text{NMSE}}=\frac{1}{N_{\mathrm{MC}}}\sum_{i=1}^{N_{\mathrm{MC}}}\text{NMSE}_i,
\end{align}
and its standard deviation 
\begin{align}
	\sigma \triangleq \sqrt{\frac{1}{N_{\mathrm{MC}}}\sum_{i=1}^{N_{\mathrm{MC}}} \left(\text{NMSE}_i - \overline{\text{NMSE}}\right)^2},
\end{align}
computed over $N_{\mathrm{MC}}{=}1000$ realizations for multisine signals and bandpass-filtered white-noise signals across several SNR values, with true offset  parameters $\Delta{=}300$~ppm and $\varepsilon{=}300$~ppm. Results are shown after one and two iterations ($m{=}1{\mid}2$) for both estimators.

\begin{table*}[ht]
	\centering
	\caption{\textbf{\textit{Example}~\ref{ex:mcperformance}. MC performance. }Evaluation of the estimators over $N_{MC}{=}1000$ realizations, evaluated after $m=1\mid 2$ iterations for different SNR values.}
	\begin{tabular}{@{}c cc cc@{}}
		\toprule
		\multirow{2}{*}{SNR (dB)} & \multicolumn{2}{c}{NM}  & \multicolumn{2}{c}{ILS} \\
		\cmidrule(lr){2-3}\cmidrule(lr){4-5}
		& $\overline{\text{NMSE}}$ ($m{=}1\mid2$) & $\sigma$ ($m{=}1\mid2$) & $\overline{\text{NMSE}}$ ($m{=}1\mid2$) & $\sigma$ ($m{=}1\mid2$) \\
		\midrule
		\multicolumn{5}{c}{\textit{Multisine}}\\
		\midrule
		20 & $2.063\times 10^{-2}\mid1.964\times 10^{-2}$ & $8.933\times 10^{-4}\mid8.666\times 10^{-4}$ & $1.992\times 10^{-2}\mid1.968\times 10^{-2}$ & $8.789\times 10^{-4}\mid8.703\times 10^{-4}$ \\
		30 & $2.966\times 10^{-3}\mid1.982\times 10^{-3}$ & $2.102\times 10^{-4}\mid8.696\times 10^{-5}$ & $2.172\times 10^{-3}\mid1.987\times 10^{-3}$ & $9.588\times 10^{-5}\mid8.717\times 10^{-5}$ \\
		\rowcolor{gray!15} 40 & $1.181\times 10^{-3}\mid1.991\times 10^{-4}$ & $1.890\times 10^{-4}\mid8.976\times 10^{-6}$ & $3.801\times 10^{-4}\mid2.019\times 10^{-4}$ & $3.445\times 10^{-5}\mid8.949\times 10^{-6}$ \\
		\midrule
		\multicolumn{5}{c}{\textit{Bandpass filtered white noise}}\\
		\midrule
		20 & $2.028\times 10^{-2}\mid1.970\times 10^{-2}$ & $1.192\times 10^{-3}\mid1.186\times 10^{-3}$ & $1.990\times 10^{-2}\mid1.973\times 10^{-2}$ & $1.197\times 10^{-3}\mid1.190\times 10^{-3}$ \\
		30 & $2.557\times 10^{-3}\mid1.988\times 10^{-3}$ & $1.612\times 10^{-4}\mid1.197\times 10^{-4}$ & $2.109\times 10^{-3}\mid1.990\times 10^{-3}$ & $1.228\times 10^{-4}\mid1.198\times 10^{-4}$ \\
		40 & $7.685\times 10^{-4}\mid1.994\times 10^{-4}$ & $1.091\times 10^{-4}\mid1.207\times 10^{-5}$ & $3.135\times 10^{-4}\mid2.005\times 10^{-4}$ & $2.360\times 10^{-5}\mid1.204\times 10^{-5}$ \\
		\bottomrule
		\vspace{-2mm}
	\end{tabular}
	\label{tab:performance}
	\begin{minipage}{\linewidth}
		\footnotesize\raggedright
		\emph{Note:} At low SNR the NMSE difference between $m{=}1$ and $m{=}2$ remains small, while at higher SNR the benefit of the second iteration becomes apparent.
	\end{minipage}
\end{table*}

At low SNR (20~dB), the NMSE reduction obtained when moving from $m=1$ to $m=2$ is small for both estimators and both signal types, indicating that the estimation accuracy is primarily limited by noise in this regime. As the SNR increases, the influence of the iterative refinement becomes more significant, and at 40~dB a clear reduction in NMSE is observed when performing the second iteration. These results further reveal a consistent pattern between the estimators: after the first iteration, ILS achieves the lowest NMSE, whereas after the second iteration NM attains a marginally smaller estimation error. The same tendency is observed in the other examples throughout the paper.

To visualize this effect, Fig.~\ref{Flo:delta_eps_iter} illustrates the convergence of the SFO and STO estimates versus iteration index $m$ for the multisine case at $\mathrm{SNR}=40$~dB, where it is seen that the dominant error reduction occurs within the first two iterations, while subsequent updates provide only minor refinement.

\begin{figure}[tbp]
	\centering
	\centering \includegraphics[scale=1.0]{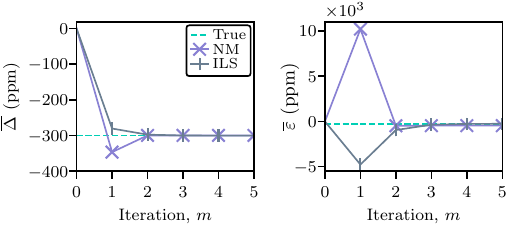}
	\caption{\textbf{\textit{Example \ref{ex:mcperformance}:} Convergence analysis. }SFO and STO precision against iteration count.}
	\label{Flo:delta_eps_iter}
\end{figure}

\example\label{ex:qam64snr}\label{ex:qam64delta}
\textbf{Sensitivity to noise and SFO value.} In this example, the estimation performance is evaluated for OFDM signals employing 64-QAM modulation with $1536$ active subcarriers out of $2048$. Two complementary aspects are analyzed: the sensitivity to noise and the robustness with respect to the SFO parameter.

First, the sensitivity to noise is analyzed. Figure~\ref{Flo:MSE} shows the NMSE and bit-error rate (BER) versus SNR for $N_{MC}{=}10\,000$ OFDM signals with fixed offsets $\Delta{=}293$~ppm and $\varepsilon{=}293$~ppm per SNR value (corresponding to ${\mid} d(n) {\mid} {\le} 0.3$). 

For both methods, at low-to-moderate SNR, the performance is dominated by noise, and a single iteration yields nearly identical NMSE and BER results as multiple iterations.\footnote{At low SNRs (e.g., below $15$~dB), the NMSE remains essentially unchanged regardless of the number of iterations, consistent with a noise-limited regime, whereas the BER continues to improve down to approximately $5$~dB.} At higher SNR values, however, the estimators become iteration limited. In particular, for SNRs above approximately $50$~dB, a single iteration results in a visible NMSE floor, whereas two iterations continue to reduce the estimation error. For SNR values around $30$~dB, the BER drops below $10^{-6}$ for ILS with two iterations, and for Newton's method with one or two iterations. In contrast, restricting ILS to a single iteration yields an iteration limited BER floor at intermediate SNRs (approximately $10^{-5}$ at $30$~dB and $10^{-6}$ at $40$~dB). 
Similar behavior was observed for the same example with 16-QAM modulation, where a second ILS iteration improves the NMSE only at high SNRs (above $45$ dB), while the BER is already below $10^{-6}$ for SNRs exceeding $25$ dB.

\begin{figure}[tbp]
	\centering
	\centering \includegraphics[scale=1.0]{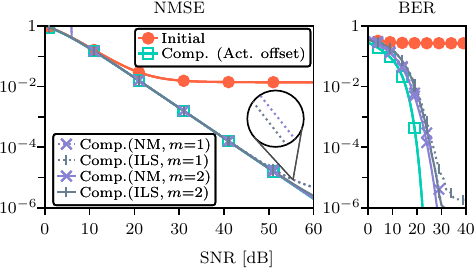}
	\caption{\textbf{\textit{Example \ref{ex:qam64snr}: }Sensitivity to noise.} NMSE versus SNR using the proposed estimators for OFDM signals.}
	\label{Flo:MSE}
\end{figure}

Next, the robustness with respect to the SFO parameter itself is evaluated. Figure~\ref{Flo:MSE_delta} shows the NMSE versus $\Delta$ for an SNR of $25$~dB, with fixed $\varepsilon{=}300$~ppm and varying $\Delta$. The results, averaged over $1\,000$ signals per $\Delta$ value, indicate that the estimation accuracy is insensitive to $\Delta$ within the considered range.\\

\begin{figure}[tbp]
	\centering
	\centering \includegraphics[scale=1.0]{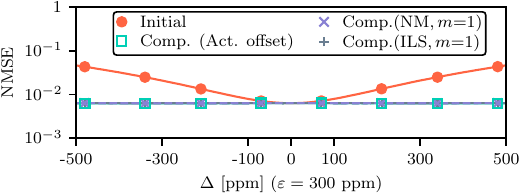}
	\caption{\textbf{\textit{Example \ref{ex:qam64delta}:} Sensitivity to SFO. }NMSE versus $\Delta$ using the proposed estimators for OFDM signals.}
	\label{Flo:MSE_delta}
\end{figure}

\begin{figure*}[tbp]
	\centering
	\includegraphics[width=\textwidth]{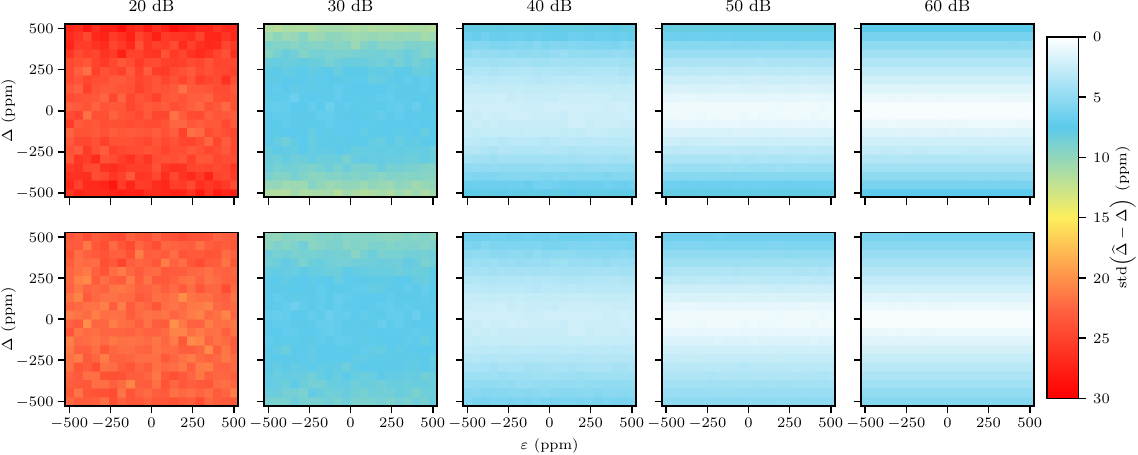}
	\caption{\textbf{\textit{Example~\ref{ex:qam64grid}: SFO estimation variability over the two-dimensional offset range. }}Standard deviation of the SFO estimation error computed over $N_{MC}=1000$ Monte-Carlo realizations, over a two-dimensional $(\Delta,\varepsilon)$ grid after a single iteration ($m{=}1$), for multiple SNR values. Top row: Newton-based estimator; bottom row: ILS estimator.}
	\label{Flo:grid}
\end{figure*}
\example\label{ex:qam64grid}
\textbf{SFO estimation variability over the two-dimensional offset range.}
To evaluate the estimation accuracy over a wide operating range, a two-dimensional grid experiment is conducted. The OFDM parameters are $1536$ active subcarriers out of $2048$, 16-QAM modulation, and $N_{\mathrm{MC}}=1000$ Monte-Carlo realizations per grid point. The grid consists of $20$ uniformly spaced values for each parameter, with $\Delta,\varepsilon \in [-500,500]~\mathrm{ppm}$, yielding $400$ distinct parameter pairs $(\Delta,\varepsilon)$. The estimation window length is fixed to $N=1000$, ensuring that $|d(n)|\le 0.5$ 
samples for all $n$ when $|\Delta|,|\varepsilon|\le500~\mathrm{ppm}$.

For each $(\Delta,\varepsilon)$ pair, both estimators are identically initialized and evaluated after a single iteration ($m{=}1$), and performance is quantified by the standard deviation of the SFO estimation error,
\begin{equation}
	\sigma_{\Delta}(\Delta,\varepsilon)
	\triangleq
	\sqrt{\frac{1}{N_{\mathrm{MC}}}
		\sum_{k=1}^{N_{\mathrm{MC}}}
		\bigl(\widehat{\Delta}^{(k)}-\Delta\bigr)^2},
\end{equation}
across all realizations, where $\widehat{\Delta}^{(k)}$ denotes the estimate obtained in the $k$-th realization.

The resulting values of $\sigma_{\Delta}(\Delta,\varepsilon)$ are reported as heatmaps in Fig.~\ref{Flo:grid}, illustrating the variability of the SFO estimation error over the $(\Delta,\varepsilon)$ grid for different SNR values and for both estimators. Both estimators exhibit very similar behavior across the grid. At low SNR ($20$~dB), the standard deviation of the SFO estimation error remains bounded within approximately $20$--$30$~ppm, which is more than one order of magnitude smaller than the total offset span. The ILS estimator exhibits a slightly lower deviation over parts of the grid.

As the SNR increases, the standard deviation of the SFO estimation error decreases significantly. In addition, the variability becomes more uniform across the grid, indicating reduced sensitivity to the specific offset pair $(\Delta,\varepsilon)$. At $30$~dB, the deviation is around $10$~ppm, and for SNRs of $40$~dB and above it falls below approximately $5$~ppm. Thus, both estimators provide consistent SFO estimation across the full offset range, with reduced variability at higher SNR.


\example\label{ex:fbnoise} 
\textbf{Operation at the edge of the Farrow design boundary.}
Next, the estimator is stressed by operating beyond the Farrow interpolator design range. The true parameters are $\Delta=450$~ppm and $\varepsilon=0.05$, which for $N{=}1024$ samples imply that $|d(n)|>0.5$ for some $n\in[0,N-1]$, i.e., the design bound is slightly violated near the end of the estimation window.\footnote{Since $N$ must be selected \emph{a priori}, the induced fractional delay ${\mid} d(n) {\mid}$ may slightly exceed the nominal design range in practice, depending on the true offset parameters.} A bandpass-filtered white noise signal with an SNR of $20$~dB is considered.

After one iteration, the estimators yield $\widehat{\Delta}=446.83$~ppm and $\widehat{\varepsilon}=0.0460$ for the ILS approach, and $\widehat{\Delta}=461.53$~ppm and $\widehat{\varepsilon}=0.0441$ for Newton’s method, both close to the true values. 
Despite violating the nominal validity region of the Farrow structure, estimation accuracy is maintained, demonstrating robustness of the proposed estimators, as most samples satisfy $|d(n)| \le 0.5$. Fig.~\ref{Flo:fbnoise} shows the corresponding spectrum before and after compensation.

\begin{figure}[tbp]
	\centering
	\centering \includegraphics[scale=1.0]{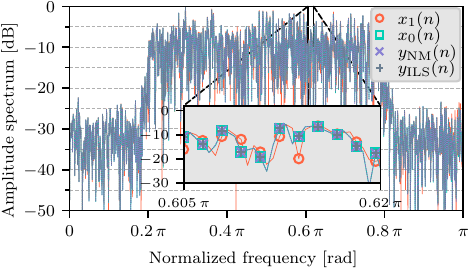}
	\caption{\textbf{\textit{Example \ref{ex:fbnoise}: Operation at the edge of the Farrow design boundary. }}Amplitude spectrum before and after compensation for a bandpass-filtered white noise signal.}
	\label{Flo:fbnoise}
\end{figure}

\example\label{ex:qam64cfocpo} 
\textbf{Robustness to other impairments (CFO and PO).}
To illustrate the robustness of the estimators in the presence of additional impairments, $5\%$ CFO and PO are introduced in an OFDM signal with $1\,536$ active subcarriers out of $2\,048$ available Nyquist subcarriers, using \mbox{64-QAM} modulation and an SNR of $30$~dB. The true SFO parameters are $\Delta{=}-300$~ppm and $\varepsilon{=}-500$~ppm, with $N{=}1024$ samples employed for the estimation (corresponding to ${\mid} d {\mid} \leq 0.31$). Using only the real part of the received signal for estimation, after one iteration, the ILS method yields $\widehat{\Delta}{=}-291$~ppm and $\widehat{\varepsilon}{=}-433$~ppm, while Newton’s method provides $\widehat{\Delta}{=}{-}303$~ppm and $\widehat{\varepsilon}{=}-2\,054$~ppm.\footnote{After the first iteration, Newton’s based algorithm misestimates $\varepsilon$, but achieves a similar NMSE as ILS method due to the accurate update of $\Delta$. After the second iteration, Newton refines the estimate to $\widehat{\varepsilon}{=}-558$~ppm.} Both estimators achieve the same compensation quality, over the entire signal (i.e., accounting for all subcarriers, before CFO/PO compensation) with a NMSE of $1.915{\times}10^{-3}$ (true-compensation reference: $1.913\times10^{-3}$).

Figure~\ref{Flo:constellation} shows the constellation diagrams of the uncompensated signal, after SFO compensation, and after subsequent CFO/PO compensation. In both cases, the final compensated constellations align well with the ideal 64-QAM grid, demonstrating effective SFO estimation and compensation even in the presence of multiple impairments. 

\begin{figure}[tbp]
	\centering
	\centering \includegraphics[scale=1.0]{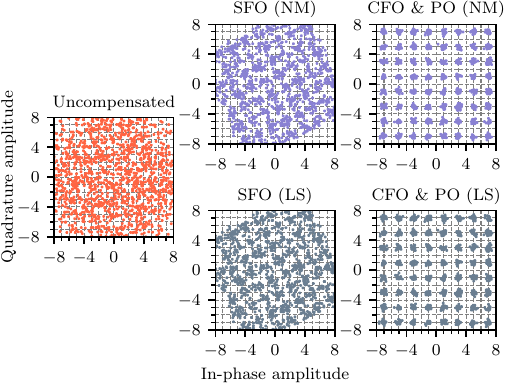}
	\caption{\textbf{\textit{Example \ref{ex:qam64cfocpo}: Robustness to other impairments (CFO and PO). }} Uncompensated signal (left), after SFO compensation (center), and after subsequent CFO/PO compensation (right).}
	\label{Flo:constellation}
\end{figure}

\subsection{VFD Filter Complexity Influence on Estimation Accuracy} \label{sec:complexity_estimation_accuracy}
As mentioned in Section~\ref{sec:compensation}, the implementation complexity of the Farrow structure is determined by the desired approximation error and the cut-off frequency.
In this subsection, we investigate how the VFD filter approximation error, and thereby its complexity, influences the estimation error.
\begin{figure*}[tbp]
	\centering \includegraphics[clip,width=\linewidth]{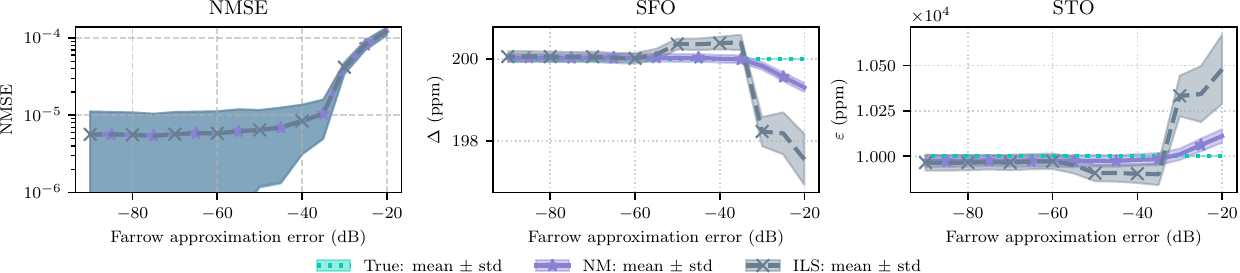}
	\caption{\textbf{\textit{Example \ref{ex:approx_error}:} Estimation accuracy vs. Farrow approximation error.} A bandpass filtered white noise signal is considered with the true offsets values being $\Delta=200$ ppm and $\varepsilon=0.01$. The curves represent the mean and standards deviation (std.) for each parameter.}
	\label{Flo:ex2_Farrow_approx_noise_param1}
\end{figure*}

\example \label{ex:approx_error} \textbf{Estimation accuracy vs. Farrow approximation error.}
In this example, the VFD filters are designed for the range of approximation errors from $-95$~dB to $-20$~dB and $\omega_c T=0.9\pi$, i.e., covering $90\%$ of the bandwidth, using the minimax based design approach \cite{johansson2003}. The parameters for the so obtained filters are listed in Table~\ref{tab:ex2_Farrow_parameters}.

Considering a real-valued filtered white noise signal, the parameters SFO and STO are estimated using the proposed methods. To investigate the practical limit of the proposed estimation methods due to the Farrow filter approximation error, we consider the case without additive noise. For $N_{\mathrm{MC}}{=}1000$ signals and $N{=}1024$ samples, the values of the estimated parameters are fixed after three iterations with the tolerance $10^{-8}$. Simulation results using different Farrow filters with parameters stated in Table~\ref{tab:ex2_Farrow_parameters} and for $\Delta=200$ ppm and $\varepsilon=0.01$  are shown in Fig.~\ref{Flo:ex2_Farrow_approx_noise_param1}. It is clearly seen that the NMSE remains essentially unchanged for the Farrow approximation errors below $-50$~dB, and increases only slightly within the range $[-50, -35]$~dB for both proposed estimation methods. Moreover, the mean trajectories of the proposed estimators, together with the regions corresponding to their standard deviations, are nearly indistinguishable from the NMSE envelope obtained when compensation is performed using the true parameter values (green lines). 
Further, for the region $[-35, -20]$~dB, the NMSE level rises significantly with the increase of the approximation error, regardless of the estimation algorithms. This degradation is also affected by the compensator itself. Although both $\varepsilon$ and $\Delta$ parameters are estimated with some deviation from the true values, the difference is more pronounced for the STO values, which indicates that SFO is less susceptible to the approximation error introduced by the Farrow filter. 
Additionally, one can observe that the ILS curves for the SFO and STO estimates are less smooth compared to the Newton method. This  behavior reveals that approximation of the higher-order terms in \eqref{eq:ILS_apprpox} results in a slightly biased LS estimation of the updates $\mathbf{\delta}_\mathbf{w}^{(m)}=[\delta_\Delta^{(m)}, \delta_\varepsilon^{(m)}]$.

\begin{table}[tbp]
	\centering
	\renewcommand{\arraystretch}{1.1}
	\caption{Approximation Error and Parameters of the Farrow Structure for $\omega_cT=0.9\pi$ ({\normalfont\textit{Example \ref{ex:approx_error}}})}
	\begin{tabular}{P{2.3cm}|P{0.8cm}|P{0.8cm}}
		\hline
		Approx. error, dB & $L$ &$N_G$\\
		\hline
		$-95$ & 7 & 62\\
		$-90$ & 7 &58\\
		$-85$ & 6 & 58\\
		$-80$ & 6 & 52\\
		$-75$ & 6 & 48\\
		$-70$ & 6 & 44\\
		$-65$ & 5 & 42\\
		$-60$ & 5 & 38\\
		$-55$ & 5 & 34\\
		$-50$ & 4 & 36\\
		$-45$ & 4 & 30\\
		$-40$ & 4 & 23\\
		$-35$ & 4 & 22\\
		$-30$ & 3 & 18\\
		$-25$ & 3 & 14\\
		$-20$ & 3 & 12\\
		\hline
	\end{tabular}
	\label{tab:ex2_Farrow_parameters}
\end{table}
\begin{figure}[t!]
	\centering \includegraphics[clip,width=\linewidth]{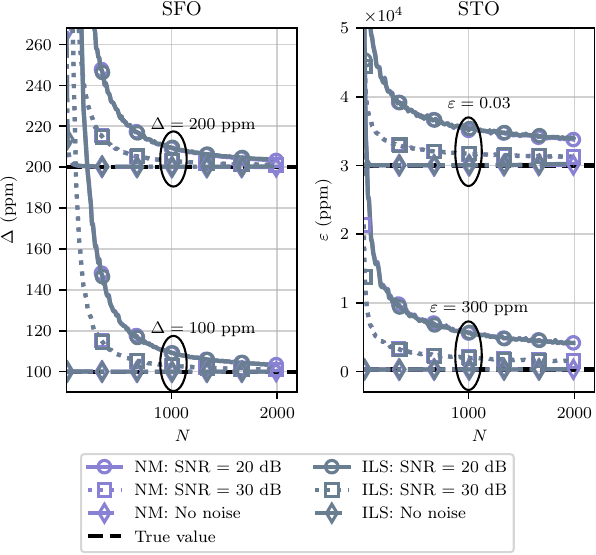}
	\caption{\textbf{\textit{Example \ref{ex:num_samples}: } Estimation accuracy vs. number of samples.} Standard deviation of the estimated parameters when a bandpass filtered white noise signal is considered. Two sets of parameter are estimated: upper curves are for ($\Delta=200$ ppm, $\varepsilon=0.03$), lower curves are for ($\Delta=100$ ppm, $\varepsilon=300$ ppm).}
	\label{Flo:ex3_batch_length_noise_param1}
\end{figure}

\subsection{Number of Samples Required for the Estimation}\label{sec:ex_number_of_samples}
As mentioned in Section~\ref{sec:complexity}, the number of samples $N$ used for estimation using the proposed algorithms needs to be chosen within the interval in \eqref{eq:N_bounds} and satisfy the trade-off between low implementation complexity and rapid convergence of the algorithms. Thus, in this subsection, the influence of $N$ on the estimation accuracy of the SFO and STO is evaluated.

\example\label{ex:num_samples} \textbf{Estimation accuracy vs. number of samples.}
Considering a real-valued bandpass filtered white-noise signal, the parameters SFO and STO are estimated using the proposed methods and different signal length $N$. Here, three cases with additive noise are evaluated: $\text{SNR}=20$~dB, $\text{SNR}=30$~dB, and the case without noise. The values of the estimated parameters are captured after two iterations. Two sets of parameters ($\Delta=200$~ ppm,~$\varepsilon=0.03$) and ($\Delta=100$~ ppm,~$\varepsilon=300$~ppm) are estimated and the standard deviation of the parameters are compared. The results are shown in Fig.~\ref{Flo:ex3_batch_length_noise_param1}, where the standard deviations are obtained based on $N_{\mathrm{MC}}{=}1000$ simulations for each length $N$ and plotted as a sum of the true parameter values and the standard deviation. One can see that the accuracy of the estimation is determined by the SNR value as expected, specifically, the higher SNR, the greater number of samples needs to be used for estimation to ensure a sufficiently low error level. Moreover, the SFO estimates are less susceptible to the additive noise compared to the STO for both methods.

\section{Proposed Estimators Using a Simplified Farrow Structure (\texorpdfstring{$L=1$}{L=1})}
\label{sec:one_branch}
To further reduce the implementation complexity of the estimation and eliminate the need for solving the nonlinear problem in \eqref{eq:argmin_cost} with the objective function in \eqref{eq:cost}, we linearize the problem by using a simplified Farrow filter with $L{=}1$ (first-degree VFD filter) for both estimation and compensation (i.e., only the branches $k=0,1$). 

In this case, the interpolated signal at the output of the simplified Farrow becomes
\begin{equation}
	y_c(n) = u_0(n) + (n\Delta + \varepsilon)u_1(n),
\end{equation}
where $u_0(n) = x_1(n) \ast g_0(n)$ and $u_1(n) = x_1(n) \ast g_1(n)$. The cost function in \eqref{eq:cost} reduces to 
\begin{equation}
	F(\Delta,\varepsilon) = \frac{1}{2}\sum_{n=1}^{M} \left[u_0(n) + (n\Delta+\varepsilon)u_1(n) - x_0(n)\right]^2.
	\label{eq:quadratic_cost}
\end{equation}

The gradient, $\mathbf{g}{=}\nabla F(\Delta,\varepsilon)$, can be written component-wise as
\allowdisplaybreaks
\begin{align}
	\frac{\partial F}{\partial \Delta} &= \sum_{n=1}^{M} \left[u_0(n) + (n\Delta+\varepsilon)u_1(n) - x_0(n)\right] \cdot n u_1(n) \notag \\
	&= \Delta\sum_{n} n^2 u_1^2(n) + \varepsilon\sum_{n} n u_1^2(n) \notag\\ 
	&+ \sum_{n} n u_1(n)\big(u_0(n)-x_0(n)\big), \\
	\frac{\partial F}{\partial \varepsilon} &= \sum_{n=1}^{M} \left[u_0(n) + (n\Delta+\varepsilon)u_1(n) - x_0(n)\right] \cdot u_1(n) \notag \\
	&= \Delta\sum_{n} n u_1^2(n) + \varepsilon\sum_{n} u_1^2(n) \notag\\
	&+ \sum_{n} u_1(n)\big(u_0(n)-x_0(n)\big),
\end{align}
which can be rewritten in matrix form as
\begin{align}
	\mathbf{g} = \underbrace{\begin{bmatrix}
			\displaystyle\sum_{n=0}^{N-1} n^2 u_1^2(n) & \displaystyle\sum_{n=0}^{N-1} n u_1^2(n) \\[8pt]
			\displaystyle\sum_{n=0}^{N-1} n u_1^2(n) & \displaystyle\sum_{n=0}^{N-1} u_1^2(n)
	\end{bmatrix}}_{\mathbf{H}}\begin{bmatrix}\Delta \\ \varepsilon\end{bmatrix} \nonumber \\ 
	+ 
	\underbrace{\begin{bmatrix}
			\displaystyle\sum_{n=0}^{N-1} n u_1(n)\big(u_0(n)-x_0(n)\big) \\[8pt]
			\displaystyle\sum_{n=0}^{N-1} u_1(n)\big(u_0(n)-x_0(n)\big)
	\end{bmatrix}}_{\mathbf{b}}.
	\label{eq:linear_gradient}
\end{align}
\paragraph{Simplified Newton's based estimator} \label{sec:simp_N} Noting from \eqref{eq:quadratic_cost}--\eqref{eq:linear_gradient} that $\mathbf{H}$ is exactly the Hessian of $F(\Delta,\varepsilon)$ and it is constant (independent of $\Delta,\varepsilon$), starting from any iterate $(\Delta^{(m)}, \varepsilon^{(m)})$, the Newton update rule in \eqref{eq:update} gives
\begin{align}
\begin{bmatrix}\Delta^{(m+1)} \\ \varepsilon^{(m+1)}\end{bmatrix} &= 
\begin{bmatrix}\Delta^{(m)} \\ \varepsilon^{(m)}\end{bmatrix} - 
\mathbf{H}^{-1}\underbrace{\left(\mathbf{H}\begin{bmatrix}\Delta^{(m)} \\ \varepsilon^{(m)}\end{bmatrix} + \mathbf{b}\right)}_{\mathbf{g}^{(m)}} \nonumber \\
&= -\mathbf{H}^{-1}\mathbf{b}, \label{eq:NS_onebranch}
\end{align}
regardless of the initial estimate.

\paragraph{Simplified ILS estimator} Applying the ILS method described in Section~\ref{sec:ILS} to the first-degree VFD filter, the cost function in \eqref{eq:quadratic_cost} can be rewritten in matrix form as
\begin{align}
	F(\Delta,\varepsilon){=}\tfrac12\|\mathbf{A}\boldsymbol{\theta}{-}\mathbf{c}_1\|_2^2
	\label{eq:ILS_cost_1}
\end{align} 
with $\mathbf{A}$ as in \eqref{eq:ILS_A} and
\begin{align}
\mathbf{c}_1 &= \begin{bmatrix}
x_0(0)-u_0(0) & \cdots & x_0(N{-}1)-u_0(N{-}1)
\end{bmatrix}^{\top},\label{eq:c_LS}
\end{align}
where the normal equations are $\mathbf{A}^\top\mathbf{A}\boldsymbol{\theta} = \mathbf{A}^\top\mathbf{c}_1$.

It is noted from \eqref{eq:ILS_update} and \eqref{eq:Q} that $\mathbf{A}^\top\mathbf{A} = \mathbf{Q}$, and for the case when $L{=}1$, $\mathbf{Q}$ coincides with $\mathbf{H}$ in \eqref{eq:linear_gradient}.
Further, from \eqref{eq:ILS_A}, \eqref{eq:linear_gradient}, and \eqref{eq:c_LS}, \mbox{$\mathbf{A}^\top\mathbf{c}_1 = -\mathbf{b}$}. It follows that the normal equation for \eqref{eq:ILS_cost_1} can be rewritten as \mbox{$\mathbf{H}\boldsymbol{\theta} = -\mathbf{b}$}, with \mbox{$\boldsymbol{\theta} = -\mathbf{H}^{-1}\mathbf{b}$}. Thus the ILS solution is obtained directly in a single step (i.e., LS) and coincides with the Newton solution in \eqref{eq:NS_onebranch}, i.e.,
\begin{equation}
	\boldsymbol{\theta}^* = -\mathbf{H}^{-1}\mathbf{b}.
	\label{eq:one_branch_solution}
\end{equation}
This explains why the implementation complexities presented in Table~\ref{tab:complexity} for both methods are the same for $L{=}1$.

\begin{remark}[Finite-step convergence for the one-branch case]
	For the quadratic cost in \eqref{eq:quadratic_cost}, the Hessian $\mathbf{H}$ is constant and positive definite.
	As a result, both Newton’s method and the ILS method converge in a single step to the unique minimizer
	$\boldsymbol{\theta}^* {=} -\mathbf{H}^{-1}\mathbf{b}$ with zero estimation error\footnote{Here, zero estimation error refers to the optimization error of the quadratic cost function. It does not account for modeling errors introduced by the Farrow approximation error or for noise-induced errors.},
	i.e., $\|\boldsymbol{\theta}^{(1)} {-} \boldsymbol{\theta}^*\| {=} 0$, independently of the initialization [i.e., the choice of $\boldsymbol{\theta}^{(0)}$]. 
	This finite-step convergence holds only for the first-degree Farrow model; higher-order branches lead to a nonlinear cost function and require iterative convergence.
\end{remark}

\subsection{Numerical Example: First-Degree Farrow Based Estimator}
\example\label{ex:qam64cfocpo1branch}
\textbf{Robustness to other impairments (CFO and PO) for the simplified estimators.}
We now consider the same scenario as in Example~\ref{ex:qam64cfocpo}, but using a simplified Farrow filter (i.e., $L{=}1$). 
In this case, the estimators are linear in the unknown parameters as in \eqref{eq:one_branch_solution}. Consequently, both methods yield identical parameter estimates, $\widehat{\Delta}=-291$~ppm and $\widehat{\varepsilon}=-433$~ppm, with a normalized mean-squared error of $\mathrm{NMSE}=2.127\times10^{-3}$. These values coincide with the estimates obtained by the ILS method after one iteration, corresponding to the initial choices $\Delta^{(0)}{=}0$ and $\varepsilon^{(0)}{=}0$ (cf. Example~\ref{ex:qam64cfocpo}).

Figure~\ref{Flo:constellation1branch} shows the constellation diagrams obtained with the first-degree estimator. Compared to the uncompensated signal, a clear improvement is observed. However, residual distortion remains, which is consistent with the slightly higher NMSE relative to the full Farrow-based estimator (i.e. $1.915{\times}10^{-3}$, also c.f. Fig.~\ref{Flo:constellation}). At this point, the estimation error is no longer the dominant performance-limiting factor; instead, the residual distortion is governed by the approximation accuracy of the interpolation itself, i.e., the first-degree VFD filter. Nevertheless, as proved by \eqref{eq:NS_onebranch} and \eqref{eq:one_branch_solution}, both estimators achieve the same performance under the simplified Farrow structure.

\begin{figure}[tbp]
	\centering
	\centering \includegraphics[scale=1.0]{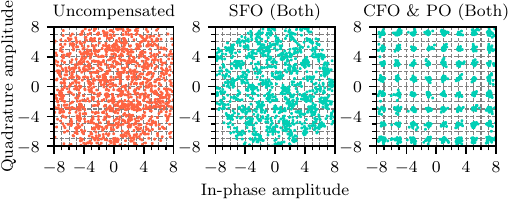}
	\caption{\textbf{\textit{Example \ref{ex:qam64cfocpo1branch}: Robustness to other impairments (CFO and PO) for the simplified estimators. }} Constellation diagrams using a first-degree Farrow. Uncompensated signal, after SFO, and after CFO/PO compensation.}
	\label{Flo:constellation1branch}
\end{figure}

\example\label{ex:qam64snr1branch}
\textbf{Sensitivity to noise comparison of the first-degree and full-Farrow based estimators.}
Figure~\ref{Flo:MSE} shows the NMSE versus SNR for $N_{\mathrm{MC}}{=}1\,000$ OFDM signals with $\Delta=488$ ($293$)~ppm and $\varepsilon=300$~ppm per SNR value (i.e., ${\mid} d{\mid} {\le} 0.5$ ($0.3$) during the estimation window $N{=}1024$ samples). For ${\mid} d{\mid} {\leq} 0.5$ ($0.3$), the simplified estimator achieves the same NMSE as the full Farrow-based estimator with ($m{=}1$) iterations for SNR values below approximately $25$~dB ($35$~dB). At higher SNRs, the NMSE of the simplified estimator is approximately one order of magnitude larger, e.g., around $10^{-3}$ ($10^{-4}$) at $60$~dB, compared to about $10^{-4}$ ($10^{-5}$) for the full Farrow-based estimator.

\begin{figure}[tbp]
	\centering
	\centering \includegraphics[scale=1.0]{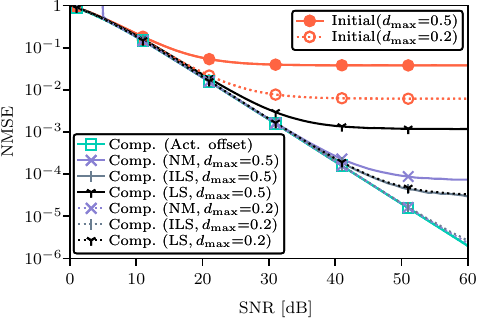}
	\caption{\textbf{\textit{Example \ref{ex:qam64snr1branch}: Sensitivity to noise comparison of the first-degree (LS) and full-Farrow based estimators (NM and ILS). }}NMSE versus SNR. } 
	\label{Flo:NMSE_simplified}
\end{figure}

\paragraph*{Approximation error for the first-degree Farrow}
The residual distortion observed in Fig.~\ref{Flo:constellation1branch}, and the difference in NMSE between the simplified and full Farrow-based estimators in Fig.~\ref{Flo:NMSE_simplified}, originate from the limited approximation capability of the simplified Farrow. To illustrate this effect, Fig.~\ref{Flo:approx1branch} reports the approximation error in the minimax sense (i.e., $\ell_{\infty}$) as a function of the normalized signal bandwidth\footnote{A normalized bandwidth of $\omega_{\max}/\pi=B$ corresponds to a signal spectrum occupying $[-B\pi, B\pi]$.} for several fractional delay values $|d|\le d_{\max}$, where $d_{\max}=0.5$ denotes the maximum supported fractional delay. As expected, the approximation error increases with both bandwidth and fractional delay, explaining the residual distortion observed after compensation in Fig.~\ref{Flo:constellation1branch}.

\begin{figure}[tbp]
	\centering
	\centering \includegraphics[scale=1.0]{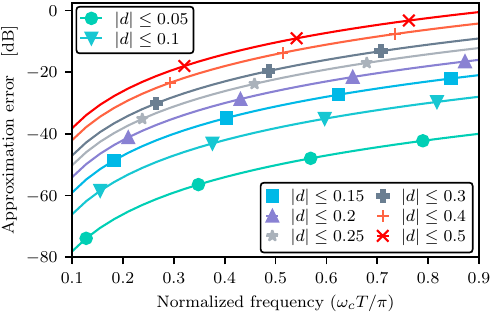}
	\caption{Approximation error of the first-degree Farrow based interpolator for different signal bandwidths and fractional delays.} 
	\label{Flo:approx1branch}
\end{figure}

\begin{remark}
	For a desired estimation-error levels (cf. Fig.~\ref{Flo:ex2_Farrow_approx_noise_param1} in Section~\ref{sec:complexity_estimation_accuracy}), the simplified estimator can be employed whenever the signal bandwidth and $d$ fall within a region where the Farrow approximation error in Fig.~\ref{Flo:approx1branch} remains below the acceptable (target) distortion level.
\end{remark}

\begin{remark}
Although $\Delta$ and $\varepsilon$ are unknown, practical systems typically operate under bounded SFO. Since the observation length $N$ is a design parameter, it may be selected to control the maximum induced delay $\{\max{\mid} d(n) {\mid}\}$ and thereby limit the approximation error of the simplified Farrow. However, reducing $N$ decreases the amount of data available for estimation and may increase the variance of the estimates, particularly at low SNR. Hence, the choice of $N$ reflects a trade-off between approximation error and noise level.
\end{remark}

\section{Conclusions\label{sec:conclusions}}
This paper presented time-domain Farrow-structure–based SFO estimation algorithms that utilize the compensator and therefore incur low implementation complexity. Two complementary estimators were developed: a Newton-based method with fast local convergence, and a reduced-complexity iterative least-squares (ILS) method with the first-degree Farrow structure utilized for the update computation. Both estimators are applicable to arbitrary bandlimited signals and can operate on a single signal component (real or imaginary part of a complex signal). Additionally, they are insensitive to other synchronization impairments such as CFO and PO.

Simulation results demonstrated that the proposed estimators provide accurate and robust joint SFO and STO estimates for a wide range of signals and SNRs. In most practical scenarios, a single iteration is sufficient, while a second iteration mainly improves performance at high SNR for the ILS method. The impact of the Farrow approximation error and estimation window length was analyzed, highlighting the trade-off between estimation accuracy and implementation complexity.

Finally, a simplified Farrow ($L{=}1$) estimator was proposed to admit a closed-form solution (i.e., no iterations) and further reduce implementation complexity, at the expense of increased estimation error at higher SNR due to the simplified Farrow structure accuracy. Overall, the results confirm that the proposed Farrow-based time-domain estimation algorithms provide an effective alternative to conventional frequency-domain SFO estimators.

\appendices
\section{Proof of Theorem \ref{teo:invertibility}}
\label{app:ILS_invertibility}


\begin{proof}
	By construction, let $\mathbf{A} \in \mathbb{R}^{M \times 2}$ be defined as in \eqref{eq:ILS_A},
	where $n_0, n_1, \dots, n_{N-1}$ are distinct sample indices and $u_1(n_l)$ are samples at the output of the Farrow subfilter corresponding to branch $k{=}1$.
	
	First, note that $\mathbf{A}^\top \mathbf{A}$ is symmetric~\eqref{eq:sym} and positive 
	semidefinite~\eqref{eq:PSD}, i.e.,
	\begin{align}
		(\mathbf{A}^\top \mathbf{A})^\top &= \mathbf{A}^\top (\mathbf{A}^\top)^\top 
		= \mathbf{A}^\top \mathbf{A}, \label{eq:sym} \\
		\mathbf{v}^\top (\mathbf{A}^\top \mathbf{A}) \mathbf{v} &= 
		(\mathbf{A}\mathbf{v})^\top (\mathbf{A}\mathbf{v}) = \|\mathbf{A}\mathbf{v}\|_2^2 {\geq} 0,
		\quad \forall \mathbf{v} \in \mathbb{R}^2. 
		\label{eq:PSD}
	\end{align}
	
	To prove positive definiteness (hence invertibility), it is enough to show that $\mathbf{v} = \mathbf{0}$ is the only solution of
	\begin{equation}
		\mathbf{A}\mathbf{v} = \mathbf{0}.
	\end{equation}
	Assume, for contradiction, that there exists a non-zero vector $\mathbf{v} = [v_1, v_2]^\top \neq \mathbf{0}$, such that $\mathbf{A}\mathbf{v} = \mathbf{0}$. This means
	\begin{equation}
		[n_l u_1(n_l)] v_1 + [u_1(n_l)] v_2 = 0, \quad l = 0, 1, \dots, N{-}1. 
		\label{eq:coordinate}
	\end{equation}
	Define $S = \{n_l : u_1(n_l) \neq 0\}$. By hypothesis, $|S| \geq 2$ and there exist $n_a, n_b \in S$ with $n_a \neq n_b$. For any $n_l \in S$, \eqref{eq:coordinate} can then be divided by $u_1(n_l) \neq 0$ to obtain
	\begin{equation}
		n_l v_1 + v_2 = 0, \quad \forall n_l \in S. 
		\label{eq:coordinate_s}
	\end{equation}
	In particular, for any pair $n_a, n_b \in S$, it yields
	\begin{align}
		n_a v_1 + v_2 &= 0, \label{eq:na} \\
		n_b v_1 + v_2 &= 0, \label{eq:nb}
	\end{align}
	and subtracting \eqref{eq:na} from \eqref{eq:nb} yields
	\begin{equation}
		(n_b - n_a) v_1 = 0.    
	\end{equation}
	Since $n_b \neq n_a$, we must have $v_1 = 0$, and substituting $v_1 = 0$ into \eqref{eq:coordinate_s} gives $v_2 = 0$, i.e., $\mathbf{v} = \mathbf{0}$, contradicting $\mathbf{v} \neq \mathbf{0}$. Therefore, $\mathbf{A}$ has full column rank, and hence $\mathbf{A}^\top \mathbf{A}$ is positive definite and invertible.
\end{proof}

\begin{small} 
	\bibliographystyle{IEEEtran}
	\bibliography{IEEEabrv, references/references} 
\end{small}

\vspace{-20 pt}
\begin{IEEEbiography}[{\includegraphics
	[width=1in,height=1.25in,clip, keepaspectratio]{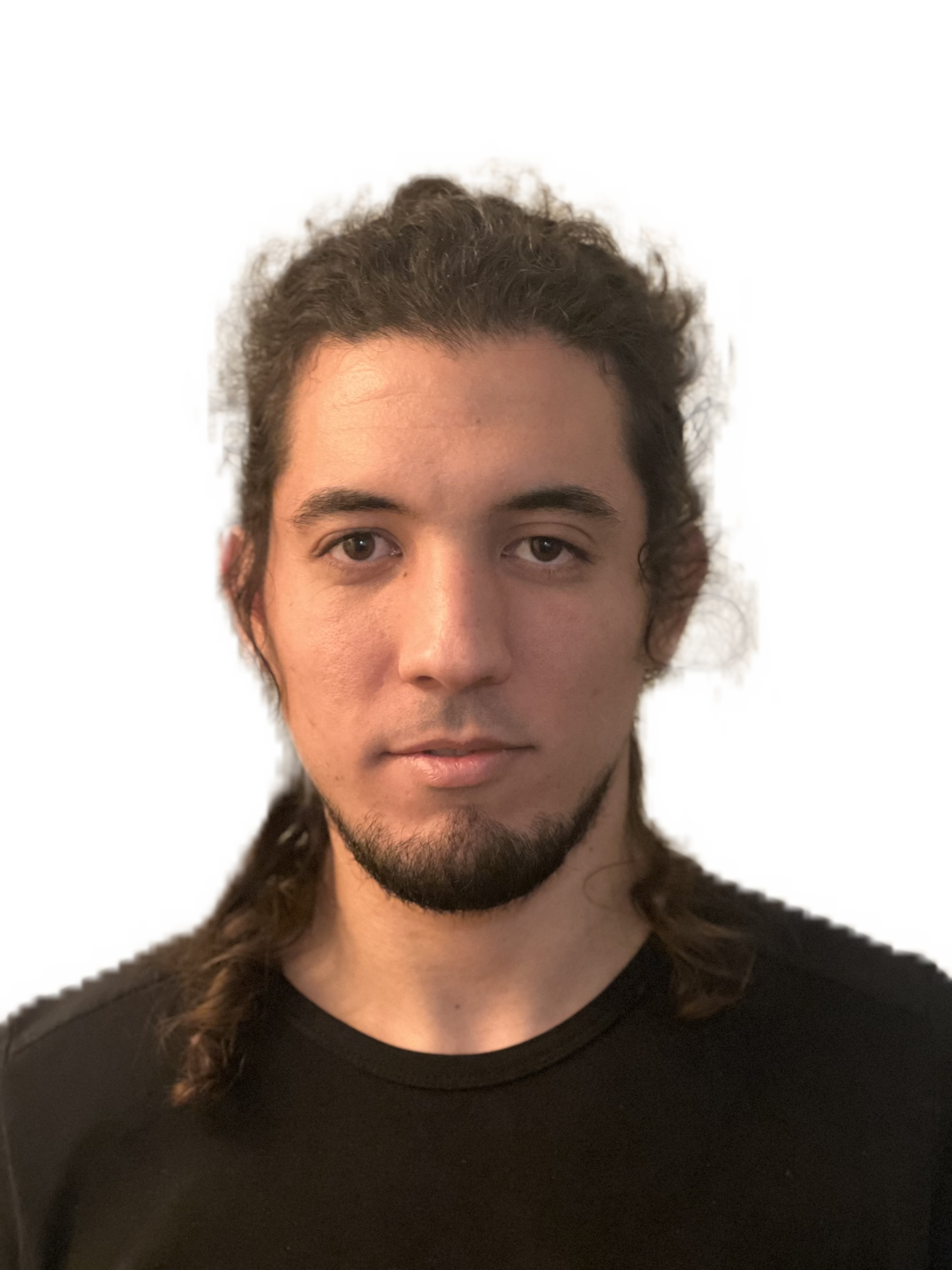}}]{Deijany Rodriguez Linares}
	(Graduate Student Member, IEEE) received the Bachelor of Science in Nuclear Engineering, a Postgraduate Diploma in Medical Physics and the Master of Science degree in Nuclear Engineering from the Higher Institute of Technologies and Applied Sciences (InSTEC), University of Havana, Cuba, in 2015, 2016 and 2018, respectively. He is currently pursuing a Ph.D. degree with the Division of Communication Systems, Department of Electrical Engineering, Link\"oping University, Sweden.
	From 2019 to 2020, he was an Associate Researcher at InSTEC, and from 2015 to 2018, he worked as a Junior Medical Physicist at the Cuban State Center for the Control of Drugs, Equipment and Medical Devices (CECMED). From 2019 to 2024, he was a Junior Associate of the Abdus Salam International Centre for Theoretical Physics (ICTP), Trieste, Italy. His research interests include signal processing, wireless communication, reinforcement learning, and nonconvex optimization.
\end{IEEEbiography}
\vspace{-20 pt}

\begin{IEEEbiography}[{\includegraphics
		[trim={7.0in 23.5in 7.0in 8.5in}, width=1in,height=1.25in,clip, keepaspectratio]{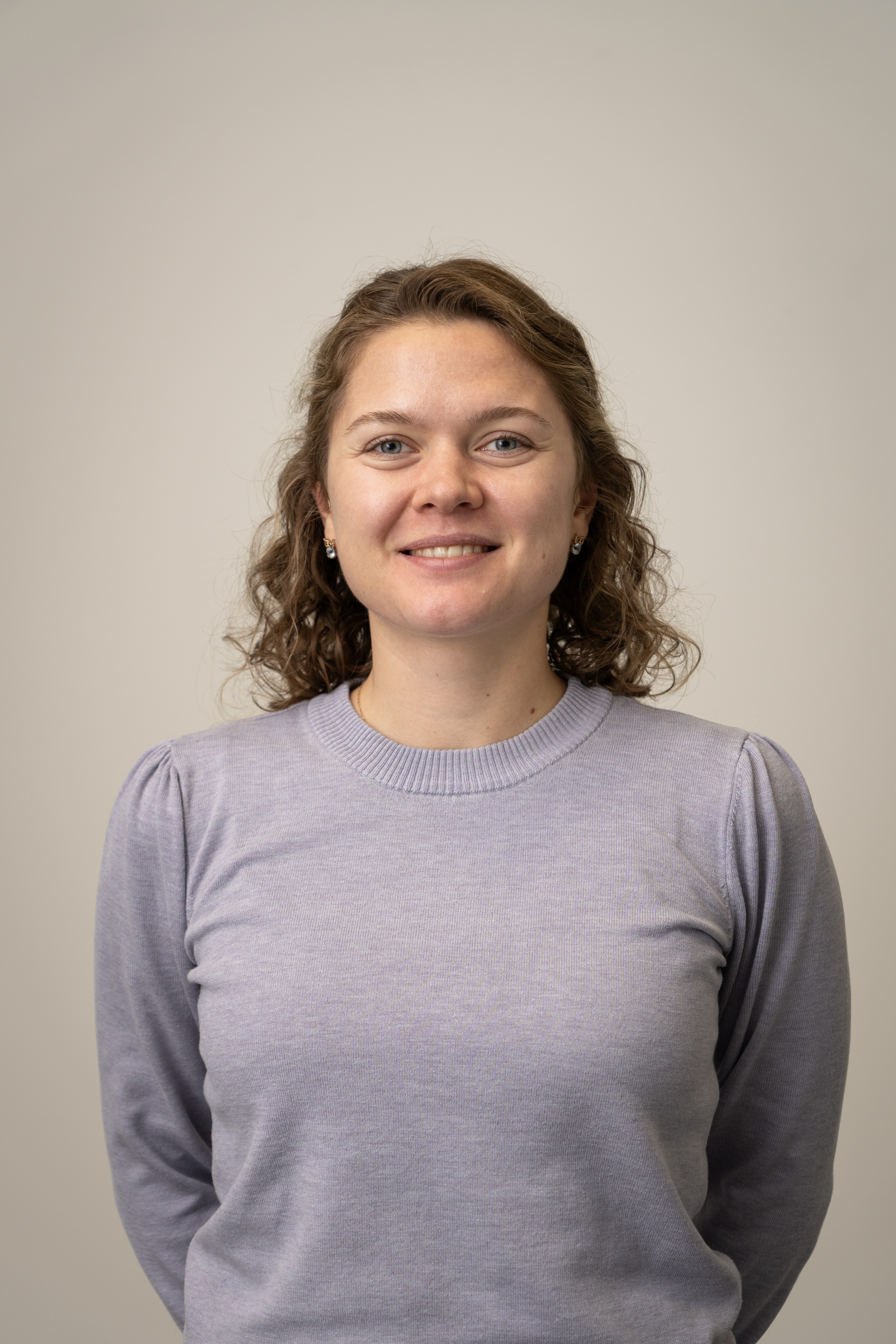}}]{Oksana Moryakova}
	(Graduate Student Member, IEEE) received the Specialist Degree in Radioelectronic Systems and Complexes (equivalent to Master of Science in Electrical Engineering) from Bauman Moscow State Technical University (\mbox{BMSTU}), Moscow, Russia, in 2020. From 2020 to 2021, she worked as a Research Engineer with Huawei Russian Research Institute. She is currently pursuing the Ph.D. degree with the Division of Communication Systems, Department of Electrical Engineering, Link\"oping University, Sweden. Her research interests include design and implementation of efficient and reconfigurable digital signal processing algorithms for communication systems, including variable digital filtering, synchronization, and linearization algorithms. 
\end{IEEEbiography}
\pdfbookmark[1]{Biography: HAkan Johansson}{bio-hakan}
\begin{IEEEbiography}[{\includegraphics
	[width=1in,height=1.25in,clip,
	keepaspectratio]{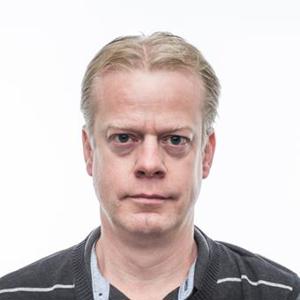}}]{H\aa kan Johansson}
	(S'97--M'98--SM'06) received the Master of Science degree in Computer Science and Engineering, and the Licentiate, Doctoral, and Docent degrees in Electronics Systems, from Link\"oping University, Sweden, in 1995, 1997, 1998, and 2001, respectively. During 1998 and 1999 he held a postdoctoral position with the Signal Processing Laboratory, Tampere University of Technology, Finland. He is currently a Professor at the Division of Communication Systems, Department of Electrical Engineering, Link\"oping University. He was one of the founders of the spin-off company Signal Processing Devices Sweden AB in 2004 (now Teledyne SP Devices). His research encompasses theory, design, and implementation of efficient and flexible signal processing systems for various purposes. He has authored or co-authored four books and some 80 journal papers and 150 conference papers. He has co-authored one journal paper and two conference papers that have received best paper awards and authored or co-authored three invited journal papers and four invited book chapters. He also holds eight patents. He served as a Technical Program Co-Chair for IEEE Int. Symposium on Circuits and Systems (ISCAS) 2017 and 2025. He has served as an Associate Editor for IEEE Trans. on Circuits and Systems I and II, IEEE Trans. Signal Processing, and IEEE Signal Processing Letters, and as an Area Editor for Digital Signal Processing (Elsevier).
\end{IEEEbiography}

\end{document}